\documentclass{article}
\usepackage{arxiv}
\usepackage{amsthm}
\usepackage[utf8]{inputenc}
\usepackage{amsmath,amssymb,amsfonts}
\usepackage{mathtools}
\usepackage{textcomp}
\usepackage{comment} % enables the use of multi-line comments (\ifx \fi) 
\usepackage{lipsum} 
\usepackage{graphicx}
\usepackage{fullpage} % changes the margin
\usepackage{amssymb, amsfonts, amsmath, amsthm}
\usepackage{bbm}
\usepackage{epsfig, latexsym, graphicx}
\usepackage{soul}
\usepackage{tikz} %add back
\usepackage{todonotes}
\usepackage{algorithm}
\usepackage{comment}
\usepackage{tabularx}
\usepackage{arxiv}
\usepackage[american]{circuitikz}
\newtheorem{theorem}{Theorem}[section]

\newtheorem{lemma}[theorem]{Lemma}

\usepackage{graphicx,epstopdf} % <- Preamble
\usepackage[caption=false]{subfig} % <- Preamble
\usepackage{algorithmic} % <- Preamble
\usepackage{enumitem}
\usepackage{appendix}
\usepackage{bigints}
\newcommand{\R}{\mathbb{R}}
\newcommand{\x}{\mathbf{x}}

\newcommand{\y}{\mathbf{y}}

\newcommand{\dfdx}{\nabla f(\x)}
\newcommand{\dfdy}{\nabla f(\y)}

\usepackage[style=ieee, backend=biber, doi=false, maxnames=2, minnames=1, url=false]{biblatex}
\theoremstyle{plain}

\theoremstyle{definition}
\newtheorem{definition}[theorem]{Definition}

\theoremstyle{remark}

\DeclareMathOperator*{\argmax}

\title{An Equivalent Circuit Approach to Distributed Optimization}
\author{Aayushya Agarwal, Larry Pileggi}
\date{April 2023}

\begin{document}

\maketitle

\section{Abstract}
\label{sec:abstract}

Distributed optimization is an essential paradigm to solve large-scale optimization problems in modern applications where big-data and high-dimensionality creates a computational bottleneck. Distributed optimization algorithms that exhibit fast convergence allow us to fully utilize computing resources and effectively scale to larger optimization problems in a myriad of areas ranging from machine learning to power systems. In this work, we introduce a new centralized distributed optimization algorithm (ECADO) inspired by an equivalent circuit model of the distributed problem. The equivalent circuit (EC) model provides a physical analogy to derive new insights to develop a fast-convergent algorithm. The main contributions of this approach are: 1) a weighting scheme based on a circuit-inspired aggregate sensitivity analysis, and 2) an adaptive step-sizing derived from a stable, Backward-Euler numerical integration. We demonstrate that ECADO exhibits faster convergence compared to state-of-the art distributed optimization methods and provably converges for nonconvex problems. We leverage the ECADO features to solve convex and nonconvex optimization problems with large datasets such as: distributing data for logistic regression, training a deep neural network model for classification, and solving a high-dimensional problem security-constrained optimal power flow problem. Compared to state-of-the-art centralized methods, including ADMM, centralized gradient descent, and DANE, this new ECADO approach is shown to converge in fewer iterations.

\section{Introduction}
Modern optimization problems are often computationally constrained by high-dimensionality and large datasets. Distributed optimization is an efficient paradigm to solve large-scale problems using all available resources by separating the optimization problem into a series of sub-problems, each of which are solved on separate computing nodes. Big data applications, such as those in machine learning, benefit from distributing  datasets across available computer resources to parallelize solving the underlying optimization problem. Similarly, distributed optimization aids high-dimensional optimization problems, such as security-constrained optimal power flow in power systems, by separating the optimization problem across multiple cores and collectively converging toward a minimum. Efficient methods to solve the distributed optimization improve efficiency and enable scaling to larger problems. 

Similar to distributed optimization, distributed simulation has been widely used to simulate large physical systems. In particular, circuit simulation uses distributed computing to simulate the behavior of large circuits approaching trillions of nodes by developing efficient distributed techniques that exploit the physical structure of the underlying circuit models \cite{white2012relaxation}. Importantly, the physical model provides a basis for deriving new distributed simulation methods. In contrast, state-of-the-art distributed optimization methods omit any domain-specific knowledge from the optimization problem for generality, but at the potential cost of efficiency. In general, distributed optimization methods tend not to utilize any structural information or domain behavior.

There is an opportunity, however, to derive insights from a physical model that can be mapped to a generalized distributed optimization one. Specifically, we map a general distributed optimization problem into a corresponding equivalent circuit (EC) model using the framework in \cite{agarwal2023equivalent}. When casted as an ODE problem, the trajectory of the optimization variables is represented by the transient response of node voltages and a steady-state solution coincides with a critical point of the objective function. We demonstrate that this physical model will facilitate the construction of a fast converging centralized distributed optimization algorithm that would otherwise be non-intuitive. 

We start by reconstructing the distributed optimization problem in terms of circuit principles to formulate a problem that is suited for distributed computing. Then we solve the corresponding EC model to steady-state (i.e., the critical point) using a distributed Gauss-Seidel (G-S) process, where each sub-problem is modeled as a partitioned sub-circuit. Importantly, we utilize circuit simulation methods to guarantee convergence of G-S and to design a novel centralized consensus algorithm. The main contributions of this consensus algorithm are 1) a weighting scheme based on an aggregate sensitivity model of each sub-problem and 2) an adaptive stable step-size selection using a Backward-Euler numerical integration scheme that allows larger step-sizes. The adaptive step-size selection routine assures numerical accuracy and guarantees Gauss-Seidel convergence. Importantly, these contributions are derived via the use of a circuit model and would not be intuitively realized without this physical model and insight.

The result is a new centralized distributed optimization algorithm, called ECADO, with a weighting scheme and adaptive, stable step-size selection method that provides fast convergence of distributed optimization problems. We compare the convergence rate of ECADO with state of the art distributed optimization methods, including DANE, centralized gradient-descent and, ADMM using a Gauss-Seidel approach. We prove the guaranteed convergence for nonconvex optimization problems and experimentally demonstrate superior convergence rate as 
compared to these existing methods.

Specifically, these comparisons are made for three possible domains: logistic regression, training distributed neural network, and solving a security constrained optimal power flow problem. For all cases and problem domains that we attempted, ECADO demonstrates faster convergence compared to  ADMM, DANE, and centralized gradient descent.

\section{Related Works}
Distributed optimization has received significant attention from optimization and control communities to leverage distributed computing resources in order to solve large-scale optimization problems. Existing distributed optimization methods have traditionally approached the problem using discrete iterative algorithms. More recently, there has been attention towards a continuous time formulation of distributed optimization to capitalize on ideas from dynamical systems for faster convergence.

\subsection{Discrete Time Algorithms}
Distributed optimization algorithms have been primarily derived in a discrete-time setting, where convergence is studied in terms pf discrete iterations. First-order optimization methods use the first-order derivatives of the objective function and have been extensively studied due to their scalability and ease of implementation. The consensus gradient-descent algorithm \cite{centralized_gradient_descent}, which performs a consensus step in the central agent that averages all local updates from each sub-problem, is the most basic first-order distributed algorithm; however, selecting an appropriate step size often requires a trade-off between speed and accuracy, such that diminishing step sizes are often employed in practice \cite{shi2015extra}.   While a complete survey of distributed optimization papers is outside the scope of this paper (refer to \cite{yang2019survey}), we highlight certain discrete-time algorithms that have demonstrated fast convergence.Several advancements have been made to improve distributed optimization methods beyond first-order techniques. For instance, EXTRA \cite{shi2015extra} utilizes past gradient information to allow larger step-sizes to accelerate convergence. In contrast, \cite{di2016next} employs a series of local convex approximations to solve nonconvex distributed optimization problems.
In the case of the centralized gradient descent (CGD) method, researchers have incorporated gradient tracking algorithms \cite{nedic2017achieving,qu2017harnessing,scutari2019distributed,xin2019distributed} to improve its linear convergence.

ADMM is a widely-used tool for solving convex optimization functions, with numerous variants developed to improve convergence rates \cite{admm_boyd2011distributed,he2016convergence,franca2018admm}. The convergence rate of ADMM has been extensively studied \cite{xu2017adaptive,wang2019global}, and it has found widespread applications in domains such as power systems optimization \cite{erseghe2014distributed,mhanna2018adaptive,yang2019survey,wang2016fully}. Despite its popularity, ADMM is highly sensitive to parameter selection \cite{nishihara2015general,ghadimi2014optimal} and can suffer from slow convergence in nonconvex settings \cite{yuejie_li2020communication}. Additionally, dual-descent methods such as ADMM \cite{admm_boyd2011distributed} and \cite{scaman2017optimal,uribe2017optimal,wai2018multi} require knowledge of the dual formulation, which may be harder to obtain. 
%A popular tool for solving convex optimization functions is ADMM \cite{admm_boyd2011distributed}, along with variants that improve on the convergence rate \cite{he2016convergence,franca2018admm}. The convergence rate of ADMM has been well-studied \cite{xu2017adaptive, wang2019global} and has been applied in various fields especially in optimizing power systems \cite{erseghe2014distributed,mhanna2018adaptive,yang2019survey,wang2016fully}. However, ADMM is very sensitivity to parameter selection \cite{nishihara2015general,ghadimi2014optimal} and can exhibit slow convergence for nonconvex applications \cite{yuejie_li2020communication}.
%Second Order

Second-order (Newton-like) methods have been employed to enhance the convergence rate by utilizing the Hessian, or an approximate Hessian, to update the state variables for a distributed optimization. The Newton-like techniques \cite{mokhtari2016network,tutunov2019distributed} attain consensus updates using the Hessian of the local objective function, which exhibit linear convergence in a small vicinity of the optimal solution. However, these methods necessitate computationally expensive evaluations of the full Hessian. Other works have avoided calculating the full Hessian, \cite{mokhtari2015approximate,dane_shamir2014communication,yuejie_li2020communication,bajovic2017newton}. Variants of \cite{dane_shamir2014communication} have been developed to solve network-aware \cite{yuejie_li2020communication}.

\subsection*{Continuous Time Algorithms}
 To gain insights into the design and convergence of optimization methods, considerable effort has been devoted to analyzing optimization trajectories using the continuous-time formulation known as gradient-flow \cite{behrman1998efficient,attouch1996dynamical,agarwal2023equivalent}. The gradient-flow approach represents optimization trajectories in a state-space, enabling the application of ideas from control systems, such as Lyapunov theory \cite{wilson2021lyapunov,polyak2017lyapunov} and feedback controls \cite{muehlebach19a}, to solve optimization problems. This approach has been extended to solve distributed optimization problems, where continuous-time algorithms have been proposed to study convergence \cite{9253684,swenson2021distributed} and design \cite{9253684,6578120,kia2015distributed,5706956,7554656} new algorithms.
%To provide insights into the design and convergence of optimization methods, there has been significant effort in analyzing optimization trajectories using a continuous-time formulation known as gradient-flow \cite{behrman1998efficient},\cite{attouch1996dynamical},[XX]. Gradient-flow develops a state-space representation of optimization trajectories and enables ideas from control systems including Lyapunov theory \cite{wilson2021lyapunov},\cite{polyak2017lyapunov} and feedback controls \cite{muehlebach19a} to be applied to solving optimization problems. This approach has been extended to solving distributed optimization, where continuous-time algorithms have been proposed to study convergence \cite{9253684},\cite{swenson2021distributed} and used to design \cite{9253684,6578120,kia2015distributed,5706956,7554656} new algorithms. 

%studying lyapunov
Lyapunov analysis, a fundamental design principle in control systems, has been applied to the dynamical system model of distributed optimization. Several existing works have leveraged Lyapunov theory to study the convergence of both new and existing distributed optimization methods \cite{sakurama_distributed,swenson2021distributed,kvaternik,Liu_second_order}. For instance, \cite{sakurama_distributed} proposes a Lyapunov stability criterion to provide general guidelines for designing feedback controllers for a distributed gradient-flow. \cite{swenson2021distributed} studies the convergence of the distributed gradient-flow using Lyapunov stability and extends the analysis to include nonsmoothness, nonconvexity, and saddle points. Meanwhile, \cite{kvaternik} demonstrates the asymptotic stability of a strongly convex objective function distributed amongst multiple agents using a Lyapunov analysis. In another example, \cite{Liu_second_order} develops a second-order distributed optimization method and applies Lyapunov theory to prove its convergence.

Proportional-integral (PI) controllers are a vital feedback mechanism in control systems. Previous research has employed the dynamical system model of distributed optimization problems to design PI controllers that drive the state towards an optimum \cite{kia2015distributed,5706956,7554656,gharesifard2013distributed, liu_small_gain}. Experimental results demonstrate that these controllers provide rapid convergence in convex settings \cite{6578120}\cite{kia2015distributed}\cite{7554656}\cite{pilloni2016discontinuous}, but have not demonstrated applicability in optimizing nonconvex objective functions. In another study, \cite{rahili2016distributed} proposes a single and double integral feedback mechanism to select an optimal trajectory, however it requires computing an inverse Hessian at each iteration. Meanwhile, \cite{liu_small_gain} implements a PI controller to achieve reference-tracking in order to optimize a distributed optimization problem with uncertain models.

 ECADO introduces a new approach based on the continuous-time formulation by utilizing an equivalent circuit model that provides novel insights and intuition for designing new optimization methods. \cite{boyd2021distributed} also leverages a circuit analogy to analyze distributed optimization; however, our work develops new methods that builds upon the circuit analogy. Based on circuit models and understanding, we created a proportional-integral (PI) feedback controller that maps proportional and integrator gains to physical elements, thereby simplifying parameter selection. ECADO leverages an averaged second-order method based on circuit principles, which not only provides greater accuracy, but is also computationally efficient. Unlike other continuous-time optimization methods, ECADO also analyzes the numerical integration properties of the continuous-time formulation, which is necessary to accurately evaluate the continuous time trajectory. To ensure accuracy and numerical stability, we developed an adaptive time-step selection algorithm that incorporates numerical integration background not previously discussed in existing distributed optimization literature. Importantly, ECADO is not restricted to convex optimization, and it demonstrates superior convergence rates for realistic, nonconvex applications.
%ECADO improves on the continuous time formulation by mapping the state-space representation to an equivalent circuit model which provides new intuition and insights to develop new methods. Based on the ECCO framework, ECADO develops a PI feedback controller, where the proportional and integrator gains are mapped onto physical elements, thus providing an intuitive method to select parameters. ECADO uses an averaged second-order method based on circuit principles that provides greater accuracy and is computationally efficient to evaluate. Lastly, unlike existing continuous-time optimization methods, ECADO analyzes the numerical integration properties of solving the continuous-time formulation. We design an adaptive time-step selection algorithm to ensure accuracy and numerical stability conditions. These conditions are developed using a numerical integration background, which, to the best of our knowledge, has not been previously discussed in existing distributed optimization literature. 

% In comparison to commonly used discrete time algorithms, including ADMM, central gradient descent [XX], and DANE [XX], we also demonstrate faster convergence rates and provide provable guarantees for convergence in nonconvex optimization problems.

 \subsection*{Paper Organization}

 Section 2 introduces the formulation of the distributed optimization problem. Section 3 presents the continuous-time formulation for the distributed optimization trajectory and defines a new problem structure that is inspired by an equivalent circuit model. Section 4 then discusses methods to solve the new state-space equations using a Gauss-Seidel approach with the conditions for guaranteeing convergence. In section 5, we introduce a new weighting scheme and an adaptive Backward-Euler step sizing method. The overall convergence guarantees and comparison of converge rates with those from discrete time algorithms are included. Finally, in Section 6, we demonstrate the efficacy of ECADO for solving convex and nonconvex optimization problems. Experimentally, we show faster rate of convergence as compared to  state-of-the-art distributed optimization methods for applied problems in machine learning and power systems.

\section{Problem Formulation}
\label{sec:problem_formulation}

We consider solving the following separable optimization problem:

\begin{equation}
    min_{\x} f(\x) 
    \label{eq:opt_obj}
\end{equation}
where $\x\in \R^n$ and $f: \R^n\to\R$.

The objective function is assumed to satisfy the following conditions:

\begin{enumerate}[wide=\parindent,label=\textbf{(A\arabic*)}]
    \item $f\in C^2$ and $\inf_{\x\in\R^n}f(\x)>-R$ for some $R>0$. \label{a1}
    \item $f$ is coercive, i.e., $\lim_{\|\x\|\to\infty} f(\x) = +\infty$. \label{a2}
    \item $\nabla^2 f(\x)$ is non-degenerate. \label{a3}
    \item (Lipschitz and bounded gradients): for all $\x,\y\in\R^n$,  $\Vert \dfdx-\dfdy\Vert\leq L\Vert\x-\y\Vert$, and $\Vert\dfdx\Vert\leq B$ for some $B>0$. \label{a4}
\end{enumerate}
\begin{definition}
We say $\x$ is a \emph{critical point} of $f$ if it satisfies $\nabla_{\x}f(\x) = \vec{0}$. Let $S$ be the set of \emph{critical points}; i.e. $S=\{\x\ |\ \dfdx = \vec{0}\}$.
\end{definition}

The coercivity and differentiability of $f$ guarantee that any minima are within the set $S$.

For the distributed set-up, the optimization problem $f(\x)$ can be separated into sub-problems by
\begin{equation}
    f(\x) = \sum_{i=1}^{m} f_i(\x), \label{eq:distributed_objective}
\end{equation}
where all $m$ separable sub-problems share a common set of variables, $\x$, and each sub-problems satisfies assumptions \ref{a1},\ref{a2}\ref{a3},\ref{a4}.

\iffalse

To design an equivalent circuit model of the optimization problem, we first model the trajectory of the optimization variables using a continuous-time state space representation, known as a scaled gradient flow [XX]. The scaled gradient-flow of the optimization problem (XX) is:
\begin{equation}
    Z(\x)\dot{\x}(t) = -\nabla f(\x(t)). \label{eq:scaled_gd_flow}
\end{equation}
The steady-state of the scaled gradient flow is defined as 
\begin{equation}
    \x^*\in \R^n s.t. \dot{\x^*}=\nabla f(\x^*) \equiv 0 \forall t>t_{ss},
\end{equation}
where $t_{ss}$ is the time at which the ODE reaches steady-state. Note the steady-state $\x^*$ is in the set of critical points of the objective function, $S=\{\x\ |\ \dfdx = \vec{0}\}$.

The scaled gradient flow provably converges to a steady-state in [XX] for objective functions that satisfy assumptions \eqref{a1,a2,a3,a4} as well as the condition that $Z$ is a diagonal matrix and that all the diagonal entries satisfy:
\begin{equation}
    Z_{ii} > 0
\end{equation}

For the distributed optimization set-up in \eqref{eq:distributed_objective}, the scaled gradient flow is defined as
\begin{equation}
    Z(\x)\dot{\x}(t) = -\sum_{i=1}^{n}\nabla f^i(\x(t)). \label{eq:gd_flow_distributed}
\end{equation}
\eqref{eq:gd_flow_distributed} represents a continuous time trajectory of the states in \eqref{eq:distributed_objective} and provides a basis for developing an equivalent circuit model as well as the new distributed optimization algorithm.

\fi

\section{Distributed Optimization as a Dynamical System}

ECADO models the trajectory of the optimization variables using a continuous-time state space model, known as a scaled gradient flow \cite{agarwal2023equivalent}. The scaled gradient-flow of the optimization problem \eqref{eq:opt_obj} is as follows:
\begin{equation}
    Z(\x)\dot{\x}(t) = -\nabla f(\x(t)).\label{eq:scaled_gd_flow}
\end{equation}
 $Z(\x)$ is a scaling matrix designed to alter the continuous-time trajectory to reach steady-state faster.

The scaled gradient-flow models the transient response of the optimization variables for a corresponding steady-state solution defined by
\begin{equation}
    \x^*\in \R^n \;\; s.t.\;\;
    \dot{\x}^*=\nabla f(\x^*) \equiv 0.
\end{equation}
 Where, $\x^*$ coincides with the set of critical points of the objective function, $S=\{\x\ |\ \dfdx = \vec{0}\}$.

The scaled gradient flow provably converges to a steady-state \cite{agarwal2023equivalent} for objective functions that satisfy assumptions \ref{a1},\ref{a2},\ref{a3},\ref{a4} and have a $Z(\x)$ that is restricted to a diagonal matrix with all diagonal entries satisfying:
\begin{equation}
    Z_{ii} > 0 \label{a5}
\end{equation}

The scaled gradient-flow equations for the separable objective \eqref{eq:distributed_objective} can be defined as
\begin{equation}
    Z(\x)\dot{\x}(t) = -\sum_{i=1}^{m}\nabla f_i(\x(t)). \label{eq:gd_flow_distributed}
\end{equation}
Where \eqref{eq:gd_flow_distributed} provides a basis for representing the new distributed optimization problem by an equivalent circuit model, as shown in Appendix \ref{appendix:separating_caps}.

We note that the scaled gradient flow, as formulated in \eqref{eq:gd_flow_distributed}, is not well-suited for distributed optimization since all of the sub-problems, $f_i$, share the same state vector, $\x$. To address this, we introduce a partitioning scheme to decompose $\x$ into $m+1$ vectors. The $m$ state vectors ($\x_i\in\Re^n \forall i\in[1,m]$) represent the state vector for each sub-problem, while the final state vector, $\x_c\in\Re^n$, represents the state vector for a centralized agent tasked with establishing consensus among all sub-problems. Importantly, this partitioning allows us to leverage insights from the dynamical system to design an effective optimization approach.

Our partitioning scheme is inspired by circuit insights described in Appendix \ref{appendix:separating_caps} and involves two steps. The first step separates the scaling matrix, $Z(\x)$, to $m+1$ scaling matrices. This establishes an independent scaling matrix, $Z_i(\x)$, for each sub-problem, and a scaling matrix, $Z_c$, for the centralized agent. In the second step, we insert a flow variable between each sub-problem and the centralized agent to separate state vectors for each agent. We integrate the flow variables into the dynamical system to provide a proportional-integral controller for added stability to a steady-state.

\subsection{Separating the Scaling Matrices}
Step one:
%The first modification to the dynamical system \eqref{eq:gd_flow_distributed} redefines the scaling matrix, $Z(\x)$ into $m+1$ separable scaling matrices by:
\begin{equation}
    Z(\x) = Z_c(\x) + \sum_{i=1}^{m} Z_i(\x), \label{eq:separating_z}
\end{equation}
Where $Z_c(\x)$ is a scaling matrix associated with the central agent, and $Z_i(\x)$ is a scaling matrix for each sub-problem ($\forall i \in [1,m]$). This modification is inspired by defining $m+1$ parallel capacitors in the equivalent circuit model as described in Appendix \ref{appendix:separating_caps}

Each scaling matrix, $Z_i$ and $Z_c$ satisfies the assumptions in \ref{a5} and can be selected to mimic various gradient-descent methods, including variable-step size gradient-descent and second-order methods. Additional details on designing a scaling matrix can be found in \cite{agarwal2023equivalent}. The resulting dynamical system is defined as:
\begin{equation}
    (Z_c(\x) + \sum_{i=1}^{m} Z_i(\x)) \dot{\x}(t) + \sum_{i=1}^{m} \nabla f_i (\x) =0. \label{eq:gd_flow_w_multiple_caps}
\end{equation}

\begin{lemma}
Separating the scaling matrices according to \eqref{eq:separating_z} does not affect the steady-state of the scaled-gradient flow ($\dot{\x}=0$).
\end{lemma}

\begin{proof}
 Steady-state is achieved when $\dot{x}=0$. The dynamical system at steady-state is defined as:
\begin{align}
     &(Z_c(\x) + \sum_{i=1}^{m} Z_i(\x)) *0 + \sum_{i=1}^{m} \nabla f_i (\x) =0.   \\
     &\implies \sum_{i=1}^{m} \nabla f_i (\x)=0.
\end{align}
Therefore, the steady-state solution $\x$ is in the set of critical points $S=\{ x| \nabla f(x)=0\}$.
\end{proof}

\subsection{Separating the state-vectors}

The second step introduces a flow variable, $I_i^L$, between the central agent and each sub-problem to separate $\x$ into $m+1$ vectors, $\x_i \in\Re^n \; \forall i \in[1,m]$ and $\x_c \in \Re^n$.
%To decompose the dynamical system into multiple state-vectors, we introduce an inductor between the central agent and each sub-problem. This step creates $m+1$ nodes in the equivalent circuit model, as detailed in Appendix \ref{appendix:adding_inductor}. 
The flow variable is incorporated into the state-space equations as follows:
\begin{align}
    &Z_c(\x_c)\dot{\x}_c(t) + \sum_{i=1}^{m} I^L_i = 0 \label{eq:gd_flow_inductor_cenrtal}\\
    &L\dot{I}_i^L = \x_c - \x_i \;\;\; \forall i\in[1,m] \label{eq:gd_flow_inductor_inductor}\\ 
    &Z_i(\x_i)\dot{\x}_i(t)+\nabla f_i(\x_i) -I_i^L = 0\;\;\; \forall i\in[1,m]. \label{eq:gd_flow_inductor_subproblem}
\end{align}

In the new set of dynamical equations, each sub-problem has a corresponding state-space equation, \eqref{eq:gd_flow_inductor_subproblem}, that is characterized by a local state-vector $\x_i$ and coupled to the central agent via a flow variable, $I_i^L \in \Re^n$.  Note in the new dynamical system, each scaling matrix, $Z_i(\x_i)$, and sub-problem gradient, $\nabla f_i(\x_i)$, are a function of the local state-vector, $\x_i$.

The state-space equations for the central agent, \eqref{eq:gd_flow_inductor_cenrtal}, provide an update for the consensus state variable, $\x_c$, that is driven by the sum of the flow variables from each sub-problem.

Inspired by the equivalent circuit representation in Appendix \ref{appendix:separating_caps}, the dynamics of the flow variables, $I_i^L$, are governed by \eqref{eq:gd_flow_inductor_inductor} with a scaling factor of $L$. $I_i^L$ acts as an actuation signal to the central agent and sub-problem that is directly proportional to the integral of the steady-state error. The state-state of the flow variables is achieved when
\begin{equation}
L \dot{I}_i^L = \x_c - \x_i \equiv 0,
\end{equation}
indicating that at steady-state the local state-vectors are identical to the central agent state-vector.
During the transient response of the system, the steady-state error can be defined as:
\begin{equation}
\varepsilon_i = \x_c - \x_i, \label{eq:steady_state_error}
\end{equation}
which approaches zero as the system converges to steady-state.

The flow variable therefore act as an integral controller that generates the following signal:
\begin{equation}
I_i^L = \frac{1}{L}\int (\x_c - \x_i) dt = \frac{1}{L} \int \varepsilon_i dt,
\end{equation}
with a controller gain of $\frac{1}{L}$. This introduces second-order effects to the dynamical system and an appropriate $L$ value, near the critical damping, ensures the dynamical system reaches steady-state faster.

These new modifications are inspired by inserting an inductor between each sub-problem (sub-circuit) and the central node of the optimization problem.. The flow variable, $I_i^L$ represents an inductor current, whose voltage-current relation is defined by \eqref{eq:gd_flow_inductor_inductor}. The EC model provides a physical analogy that justifies the new set of dynamical equations. Further details on inserting the inductor into the EC model are provided in Appendix \ref{appendix:adding_inductor}.

%The modifications introduced to the dynamical system offer a partitioning scheme, where each sub-problem is defined by a local state vector, $\x_i$, while the central agent provides the consensus. Note that each scaling matrix and sub-problem depends on its corresponding state-vector, $\x_i$, thus effectively reducing the coupling among sub-problems. The steady-state of this new dynamical system is solved by ECADO.

\iffalse
Upon reaching steady-state, the inductor current becomes constant, guaranteeing that:
\begin{equation}
L \dot{I}_i^L = \x_c - \x_i \equiv 0.
\end{equation}
During the transient response of the system, the steady-state error can be evaluated as:
\begin{equation}
\varepsilon_i = \x_c - \x_i,
\end{equation}
which approaches zero as the system converges to steady-state. The inductor current serves as an actuation signal to the central agent that is directly proportional to the integral of the steady-state error. The inductor is a physical model that introduces an integral controller into the system that generates the following signal:
\begin{equation}
I_i^L = \frac{1}{L}\int (\x_c - \x_i) dt = \frac{1}{L} \int \varepsilon_i dt,
\end{equation}
with a controller gain of $\frac{1}{L}$. Choosing an appropriate $L$ value, near the critical damping, ensures the stability of the steady-state in the dynamical system.
\fi

\begin{theorem}

The modified gradient-flow equations in \eqref{eq:gd_flow_inductor_cenrtal}-\eqref{eq:gd_flow_inductor_subproblem} achieve the same steady-state as \eqref{eq:gd_flow_distributed} that coincide with the critical points of \eqref{eq:distributed_objective}.

\end{theorem}

\begin{proof}

The steady-state of the dynamical system \eqref{eq:gd_flow_inductor_cenrtal}-\eqref{eq:gd_flow_inductor_subproblem} is defined by zero time-derivative values, namely, $\dot{\x_i}(t)=0$ and $L\dot{I}_i^L =0$. The corresponding set of equations that characterizes the system in steady-state is:
\begin{align}
    \sum_{i=1}^{n} I_i^L &=0 \\
    \x_c - \x_i & =0 \forall i\in n \\
    \nabla f_i(\x_i) -I_i^L &=0 \forall i \in n
\end{align}
The equations above can be rearranged to describe the staedy-state condition as:
\begin{equation}
    \sum_{i=1}^{n} \nabla f^i(\x_c) =0,
\end{equation}
which is defined as a critical point of the separable objective function. 
\end{proof}

\section{Solving the Distributed EC Model}

These equivalent circuit inspired modifications generate a new structure to the problem that ECADO exploits for distributed computing. The new structure of the scaled gradient-flow in \eqref{eq:gd_flow_inductor_cenrtal},\eqref{eq:gd_flow_inductor_inductor},\eqref{eq:gd_flow_inductor_subproblem} is illustrated in a matrix form below:
\begin{equation}  
\begin{bmatrix}
Z(\x_c) & 0 & 0 &0& 0&0&\ldots \\
0 & L & 0 & 0 &0&0& \ldots \\
0 & 0 & L & 0 & 0&0&\ldots \\
0 & 0 & 0 & \ddots & 0&0& \ldots \\
0 & 0 & 0 & 0 & Z_1(\x_1) & 0 & \ldots \\
0 & 0 & 0 & 0 & 0& Z_2(\x_2) & \ldots \\
0 & 0 & 0 & 0 & 0 & 0 & \ddots \\
\end{bmatrix} 
\begin{bmatrix}
    \dot{\x}_c (t) \\
    \dot{I}_1^L(t) \\
    \dot{I}_2^L(t) \\
    \vdots \\
    \dot{\x}_1(t)\\
    \dot{\x}_2(t)\\
    \vdots
\end{bmatrix}
 + 
 \begin{bmatrix}
     0 & 1 & 1 &\ldots &0 &0 & \ldots \\
     1 & 0 &0 &\ldots & -1& 0& \ldots\\
     1& 0 & 0 & 0 &\ldots & -1& \ldots\\
     \vdots & 0 & 0 & \ldots & 0 &0 &-1 \\
     0 & -1 &0 & \ldots & 0& 0& 0\\
     0 & 0 &-1 & \ldots & 0& 0& 0\\
     0 & 0 &0 & \ddots & 0& 0& 0\\   
 \end{bmatrix}    \begin{bmatrix}
        \x_c(t) \\
        I_1^L (t) \\
        I_2^L (t) \\
        \vdots \\
        \x_1(t) \\
        \x_2(t) \\
        \vdots \\
    \end{bmatrix}
+
\begin{bmatrix}
    0 \\
    0 \\
    0 \\
    \vdots \\
    \nabla f_1(\x_1) \\
    \nabla f_2(\x_2) \\
    \vdots 
\end{bmatrix}
\end{equation}

The matrix structure has a positive block-diagonal left hand matrix (which we denote as $Z-$matrix) and weak coupling between the central agent and each sub-problem. 
We further abstract the problem as:
\begin{equation}
    \begin{bmatrix}
        \mathbf{Z}_c & 0 & 0 &0 & \ldots \\
        0 & Z_1(\x_1) & 0& 0 & \ldots \\
        0 & 0 & Z_2(\x_2) & 0 & \ldots \\
        0 & 0 & 0 & \ddots & 0 \\
        0 & 0 & 0 & 0 & Z_m(\x_m)
    \end{bmatrix}
    \begin{bmatrix}
        \dot{y}_c(t) \\
        \dot{\x}_1(t) \\
        \dot{\x}_2(t) \\
        \vdots \\
        \dot{\x}_m(t) \\
    \end{bmatrix}
+
    \begin{bmatrix}
        A & -1 & -1 & -1 & -1 \\
        [-1]_{2,2} & 0 & 0 & 0 & 0\\
        [-1]_{3,3} & 0 & 0 & 0 & 0\\
        \vdots & 0 & 0 & \ddots & 0\\
        [-1]_{m,m} & 0 & 0 & 0 & 0\\
    \end{bmatrix}
    \begin{bmatrix}
        y_c(t) \\
        \x_1(t) \\
        \x_2(t) \\
        \vdots \\
        \x_m(t) \\        
    \end{bmatrix}
 +
    \begin{bmatrix}
        0 \\
        \nabla f_1(\x_1) \\
        \nabla f_2(\x_2) \\
        \vdots \\
        \nabla f_m(\x_m) \\        
    \end{bmatrix}, \label{eq:matrix_distributed_ec}
\end{equation}
where $y_c = [x_c, I_1^L, I_2^L, \ldots, I_m^L]^T$ represents the states of the central agent along with the flow variables. $\mathbf{Z_c}\in \Re^{n+nm, n+nm}$ in the right-hand $Z-$matrix and the matrix $A$ is defined as:
\begin{equation}
    \mathbf{Z_c} = \begin{bmatrix}
        Z_c(\x_c) & 0 & 0 & 0 \\
        0 & L & 0 & 0 \\
        0 & 0 & \ddots & 0 \\
        0 & 0 & 0 & L
    \end{bmatrix}, \;\;\;\;  A = \begin{bmatrix}
        0 & 1 & 1 & \ldots & 1 \\
        1 & 0 & 0 & \ldots & 0 \\
        1 & 0 & 0 & \ldots & 0 \\
        \vdots & 0 & 0 & \ldots & 0 \\
        1 & 0 & 0 & \ldots & 0 \\
    \end{bmatrix},
\end{equation}
where $1$ refers to an identity matrix.

To solve for the critical point of the distributed optimization problem, ECADO simulates the transient response of \eqref{eq:matrix_distributed_ec} by marching in discrete time-steps until the system reaches steady-state. ECADO solves for the state of the entire dynamical system as:
\begin{multline}
\label{eq:ode_integral}
        \begin{bmatrix}
        y_c(t+\Delta t) \\
        \x_1(t+\Delta t) \\
        \x_2(t+\Delta t) \\
        \vdots \\
        \x_m(t+\Delta t) \\        
    \end{bmatrix} = \begin{bmatrix}
        y_c(t) \\
        \x_1(t) \\
        \x_2(t) \\
        \vdots \\
        \x_m(t) \\        
    \end{bmatrix} - \\ \bigint_{t}^{t+\Delta t}     \begin{bmatrix}
        \mathbf{Z}_c^{-1} & 0 & 0 &0 & \ldots \\
        0 & Z_1^{-1}(\x_1) & 0& 0 & \ldots \\
        0 & 0 & Z_2^{-1}(\x_2) & 0 & \ldots \\
        0 & 0 & 0 & \ddots & 0 \\
        0 & 0 & 0 & 0 & Z_m^{-1}(\x_m)
    \end{bmatrix} \left(    \begin{bmatrix}
        A & -1 & -1 & -1 & -1 \\
        [-1]_{2,2} & 0 & 0 & 0 & 0\\
        [-1]_{3,3} & 0 & 0 & 0 & 0\\
        \vdots & 0 & 0 & \ddots & 0\\
        [-1]_{m,m} & 0 & 0 & 0 & 0\\
    \end{bmatrix}
    \begin{bmatrix}
        y_c(t) \\
        \x_1(t) \\
        \x_2(t) \\
        \vdots \\
        \x_m(t) \\        
    \end{bmatrix}
 +
    \begin{bmatrix}
        0 \\
        \nabla f_1(\x_1) \\
        \nabla f_2(\x_2) \\
        \vdots \\
        \nabla f_m(\x_m) \\        
    \end{bmatrix} \right) dt
\end{multline}

The integral on the right-hand side often does not have a closed-form solution and is approximated using numerical integration methods \cite{agarwal2023equivalent}. For example, we can apply an explicit Forward-Euler integration to solve for the series of time points:
\begin{equation}
        \begin{bmatrix}
        y_c(t+\Delta t) \\
        \x_1(t+\Delta t) \\
        \x_2(t+\Delta t) \\
        \vdots \\
        \x_m(t+\Delta t) \\        
    \end{bmatrix} = \Delta t     \begin{bmatrix}
        \mathbf{Z}_c^{-1} & 0 & 0 &0 & \ldots \\
        0 & Z_1^{-1}(\x_1) & 0& 0 & \ldots \\
        0 & 0 & Z_2^{-1}(\x^2) & 0 & \ldots \\
        0 & 0 & 0 & \ddots & 0 \\
        0 & 0 & 0 & 0 & Z_n^{-1}(\x_n)
    \end{bmatrix} \left(    \begin{bmatrix}
        A & -1 & -1 & -1 & -1 \\
        [-1]_{2,2} & 0 & 0 & 0 & 0\\
        [-1]_{3,3} & 0 & 0 & 0 & 0\\
        \vdots & 0 & 0 & \ddots & 0\\
        [-1]_{m,m} & 0 & 0 & 0 & 0\\
    \end{bmatrix}
    \begin{bmatrix}
        y_c(t) \\
        \x^1(t) \\
        \x^2(t) \\
        \vdots \\
        \x^n(t) \\        
    \end{bmatrix}
 +
    \begin{bmatrix}
        0 \\
        \nabla f^1(\x^1) \\
        \nabla f^2(\x^2) \\
        \vdots \\
        \nabla f^n(\x^n) \\        
    \end{bmatrix} \right).
\end{equation}
The Forward-Euler step is identical to performing gradient-descent iterations with a fixed step size equal to $\Delta t$ \cite{agarwal2023equivalent}.

When dealing with large datasets or high-dimensional problems, computing the integral in \eqref{eq:ode_integral} on a single compute node is often computationally intractable. Therefore, we propose a distributed approach that is based on Gauss-Seidel (G-S). G-S updates the state for each variable, at  each time point, by solving each sub-problem independently on a separate compute node. Subsequently, the updates of each sub-problem are communicated to the central agent, which performs a consensus step.

Algorithm \ref{basic-gauss-seidel} outlines our general Gauss-Seidel approach. By dividing the problem into smaller sub-problems and distributing the computation, our method is capable of scaling to accommodate larger datasets and higher dimensional problems.

The algorithm iteratively solves for the system states $[\x_c, I_i^L,\x_i]$ over time-windows of $[t_1, t_2]$ until the system reaches a steady-state characterized by $\dot{\x}_c=0$, thereby satisfying the outer while loop condition on line 6. 

In between discrete timepoints, our G-S approach (performed on lines 9-13) decouples each sub-problem from the central agent by assuming the flow variables, $I_i^L$, are constant. Each sub-problem is then solved independently over the time-window $[t_1,t_2]$ using the values of the flow variables, $I_L^i$, from the previous G-S iteration (line 10). The updates from each sub-problem, $\x_i^{k+1}$, are then used by the central agent to solve for its state-variables, $\x_c^{k+1}$, and the new flow variables, $I_i^{L^{k+1}}$, according to \eqref{eq:matrix_distributed_ec} (line 11-13).

The Gauss-Seidel process (line 10-13) continues until convergence, $|\x_c^{k+1} - \x_c^{k}|^2 = 0$ .

\begin{algorithm}
\caption{Gauss-Seidel for Distributed EC Model}
\label{basic-gauss-seidel}
\textbf{Input: } $\nabla f(\cdot)$,$\x(0)$, $\Delta T>0$
\begin{algorithmic}[1]
\STATE{$\x_c \gets \x(0)$}
\STATE{$\x_i \gets \x(0)$}
\STATE{$I_i^L \gets 0$}
\STATE{$t_1 \gets 0$}
\STATE{$t_2 \gets \Delta T$}

\STATE{\textbf{do while} $\|\dot{\x}_c\|^2 > 0$}
%\STATE{\hspace*{\algorithmicindent}$t\gets 0$}
\STATE{\hspace*{\algorithmicindent}$\x_c^k \gets \x_c^{k+1}$}
\STATE{\hspace*{\algorithmicindent}$\x_i^{k} \gets \x_i^{k+1}$}
\STATE{\hspace*{\algorithmicindent}\textbf{do while} $\|\x_c^{k+1} - \x_c^{k}\|^2 > 0$}
\STATE{\hspace*{\algorithmicindent}\hspace*{\algorithmicindent}Parallel Solve for $\x_i^{k+1}(t)$: $Z(\x_i^{k})\dot{\x}_i^{k+1}(t) = -\nabla f(\x_i^{k+1}(t)) + I_i^{L^k}(t) \;\;\; \forall i \in [1,n]\;\;\; \forall t\in[t_1,t_2] $}

\STATE{\hspace*{\algorithmicindent}\hspace*{\algorithmicindent}Solve for $\x_c^{k+1}(t), I_i^{L^{k+1}}(t)$: }
\STATE{\hspace*{\algorithmicindent}\hspace*{\algorithmicindent}\hspace*{\algorithmicindent} $Z_c\dot{\x_c}^{k+1}(t) +\sum_{i=1}^n I_i^{L^{k+1}}(t) =0 \;\;\; \forall t\in[t_1,t_2] $}
\STATE{\hspace*{\algorithmicindent}\hspace*{\algorithmicindent}\hspace*{\algorithmicindent} $L\dot{I}_i^{L^{k+1}}(t)=\x_c^{k+1}(t)-\x_i^{k+1}(t) \;\;\; \forall t\in[t_1,t_2]$}
\STATE{\hspace*{\algorithmicindent} $t_1 += \Delta T$}
\STATE{\hspace*{\algorithmicindent} $t_2 += \Delta T$}
\RETURN $\x_c$
\end{algorithmic}
\end{algorithm}

%The new structure of the problem clearly defines a partitioning whereby each sub-problem can be solved independently using the values of the inductor current, $I_L^i$, from the previous iteration.

%The central agent then uses the updates from each sub-problem, $\x_i^{k+1}$ to solve for the central agent state-variables, $\x_c^{k+1}$ and the new inductor currents, $I_i^{L^{k+1}}$ according to  \eqref{eq:matrix_distributed_ec}.

 Part of the novelty of our ECADO approach is defining the new consensus algorithm to solve for the central agent states. Our novel approach for the central agent is based on two contributions:  1) an aggregate sensitivity model of each sub-problem, and 2) a Backward-Euler numerical integration. We next detail the general methodology to solve each sub-problem, then derive the steps for solving for the central agent states.

\subsection{Solving Sub-Problem}

In Algorithm \ref{basic-gauss-seidel}, each sub-problem is solved independently using a constant value for the flow variables from the previous iteration, $I_L^{i^k}$. The differential equation characterizing the sub-problem is:
\begin{equation}
    Z(\x_i^{L^{k+1}})\dot{\x}_i^{k+1}(t) + \nabla f_i(\x_i^{k+1}) + I_i^{L^k}(t) = 0 \;\;\; \forall t \in [t_1,t_2], \label{eq:sub_problem_ode}
\end{equation}

where $I_i^{L^k}(t)$ is known from the previous G-S iteration $\forall t\in[t_1,t_2]$. 

 The state at each time-point, $\x_i(t+\Delta t)$, is determined by:
\begin{equation}
    \x_i^{k+1}(t+\Delta t) = \x_i^{k+1}(t) - \int_{t}^{t+\Delta t} Z_i^{-1}(\x_i^{k+1}) \left(-I_i^{L^{k+1}}(t) + \nabla f_i(\x_i^{k+1}(t))\right) dt \;\;\; \forall t\in[t_1,t_2]. \label{eq:sub_problem_int}
\end{equation}

To solve for the states in \eqref{eq:sub_problem_int}, two steps are required: 1) designing the scaling matrix, $Z(\x_i)$, and 2) selecting a numerical integration method to approximate the integral in \eqref{eq:sub_problem_int}. Users have the flexibility to select $Z(\x_i)$ and numerical integration method to achieve their desired optimization effects. This is similar to selecting an optimization method to solve each sub-problem.

For instance, using a unit scaling matrix, $Z_i=I$, mimics an uncontrolled gradient-descent, while $Z_i = \nabla^2f_i$ uses the Hessian to model the trajectory of a Newton-Raphson approach. Moreover, users can incorporate momentum-like effects by selecting an explicit Runge-Kutta numerical integration method \cite{agarwal2023equivalent}. For additional information on selecting $Z_i$ and the integration method, refer to \cite{agarwal2023equivalent}.

\subsection{Solving the Central Agent States}

During the G-S iteration, each sub-problem communicates local updates, $\x_i^{k+1}$, to the central agent. The central agent updates the flow variables, $I^L$, and central agent state variables, $\x_c$, are updated by solving the following differential equation:
\begin{align}
    &Z_c \dot{\x_c^{k+1}}(t) + \sum_{i=1}^n I_i^{L^{k+1}} =0 \label{eq:central_agent_ode1} \\
    & L \dot{I}_i^{L^{k+1}} = \x_c^{k+1} - \x_i^{k+1}, \label{eq:central_agent_ode}
\end{align}
where $\x_i^{k+1}$ are the updated state-variables of each sub-problem for the most recent GS iteration. The state vector, $\x_c$, provides a consensus state as is typically used in centralized distributed optimization methods.%The central agent solves for a consensus state-variable vector, $\x_c$ in \eqref{eq:central_agent_ode1},\eqref{eq:central_agent_ode} using the updates from all distributed sub-problems, similar to consensus algorithms used in centralized distributed optimization algorithms.

ECADO utilizes a two-step process to solve for the central agent state variables in the differential equations described above. In the first step, ECADO introduces an aggregate sensitivity model to the state-space equations \eqref{eq:central_agent_ode}. This provides a weighting scheme based on sensitivity that drives the updates in the central agent states. In the second step, ECADO solves the central agent states using a stable Backward-Euler integration method. This method facilitates faster convergence and has fewer restrictions on step sizes compared to other integration methods. These two steps are derived based on the equivalent circuit model (Appendix \ref{sec:partitioning_ec}) that provides a robust approach to solving the optimization problem.

\subsubsection{Aggregate Sensitivity Model}

During each time window, $[t_1,t_2]$, the central agent receives updates for the state-variables of all sub-problems, $\x_i^{k+1}$. These updates represents each sub-system by a constant value of $\x_i^{k+1}$ during this window.
The constant-valued model of each sub-system, however, does not capture the effect of changes in $\x_c$ and $I_i^L$ on $\x_i$. To address this, we introduce a \emph{linear sensitivity model} to capture the sensitivity of each sub-system to changes in the central agent state variables during the consensus step. The sensitivity model is derived via a commonly applied circuit concept known as the \emph{Thevenin model} (further details are provided in Appendix \ref{sec:partitioning_ec}). This representation consists of a constant value, $\x_i^{th^{k+1}}$, and a linear sensitivity term, $R_i^{th}$, to represent a sub-circuit, which in this case is a sub-problem, as follows:
\begin{equation}
     R_i^{th}I_i^{L^{k+1}} + \x_i^{th^{k+1}},
\end{equation}
where the constant value is defined as:
\begin{equation}
    \x_i^{th^{k+1}} = \x_i^{k+1} - R_i^{th}I_i^{L^{k}}.
\end{equation}

$R_i^{th}$ models the linear sensitivity of the state variable, $\x_i$, with respect to the flow variable, $I_i^L$, and is calculated as:
\begin{equation}
    R_i^{th} = \frac{\partial \x_i}{\partial I_i^L} = (\frac{Z_i}{\Delta t}+\frac{\partial \nabla f_i(\x_i)}{\partial \x_i})^{-1}, \label{eq:Rth}
\end{equation}
where $\frac{\partial \nabla f_i(\x_i)}{\partial \x_i}$ is the Hessian of the sub-problem, $f_i$.
The complete derivation for $R_i^{th}$ is provided in Appendix \ref{sec:deriving_rth}.

The sensitivity model enhances the consensus step by enabling a better prediction of how each sub-problem will respond to changes in the central agent variables ($\x_c$ and $I_i^L$). This results in a more effective convergence toward steady-state. The central agent dynamics now include the linear sensitivity model:
\begin{align}
    Z_c\dot{\x_c}^{k+1}(t) + \sum_{i=1}^n I_i^{L^{k+1}} =0 \\
    L\dot{I}_i^{L^{k+1}} = \x_c^{k+1} - ( I_i^{L^{k+1}} R_i^{th} + \x_{i}^{{th}^{k+1}}) \label{eq:central_chord_ode_inductor}\\
    \x_{i}^{{th}^{k+1}} = \x_i^{k+1} - I_i^{L^k} R_i^{th} \label{eq:central_chord_ode}
\end{align}

The linear sensitivity model for our ECADO approach is based on these sensitivities,  , $R_i^{th}$, for each $ith$ sub-system to generate the weighting that drives the dynamics of the flow variable, $I_i^L$, \eqref{eq:central_chord_ode_inductor}. Unlike other Hessian-based weighting schemes such as DANE \cite{dane_shamir2014communication}, our linear sensitivity model also incorporates the effect of the integration process with the term, $\frac{Z_i}{\Delta t}$, in \eqref{eq:Rth}. Specifically, the numerical integration term in our approach establishes a set of nonlinear equations (known as a companion model \cite{lawrence1995electronic}) that depends on the nonlinear term,  $\nabla f_i$, integration method and time-step, $\Delta t$. This results in the overall linear sensitivity of ECADO incorporating aspects of the integration method to provide an accurate model of the sub-problem.

 ECADO integrates the sensitivity model into the dynamics of the flow variables via \eqref{eq:central_chord_ode_inductor}. This results in a steady-state error, $\varepsilon_i$, (defined in \eqref{eq:steady_state_error}) of:

\begin{equation}
    \varepsilon_i = \x_c - \x_i = L\dot{I}_i^{L^{k+1}} + I_i^{L^{k+1}}R_i^{th}. 
\end{equation}
This equation demonstrates that the added linear sensitivity  acts as a proportional controller with a controller gain of $R_i^{th}$. 

%We now have designed a new model for distributed optimization that leverages knowlege from circuit design and where each parameter, such as sensitivity, can be carefully selected using a circuit interpretation.

\paragraph{Averaged Aggregate Sensitivity Model}
Computing the linear sensitivity in \eqref{eq:Rth} would require evaluating the Hessian at each G-S iteration, which would likely incur a computational cost that would outweigh any runtime benefit. Therefore, we employ a successive chord approach instead of a Newton-based one for the inner loop nonlinear iterations (analogous to successive chords and lazy Newton approaches applied in in circuit simulation \cite{lawrence1995electronic}) that approximates the Hessian using an averaged Hessian value. For our successive chords approach, we pre-compute for all of the sub-problems an averaged Hessian across a range of values for $\x_i$:
\begin{equation}
    \bar{H} = \frac{1}{p} \sum_{j=1}^{p} \frac{\partial \nabla f_i(\x_i^j)}{\partial \x_i^j},
\end{equation}

The linear sensitivity model can then use the pre-computed, average Hessian to reduce runtime cost as:
\begin{equation}
    \bar{R}_i^{th} = (\frac{Z_i}{\Delta t} + \bar{H})^{-1}. \label{eq:average_thevenin_resistance}
\end{equation}

%The Gauss-Seidel iteration is a fixed-point iteration technique that is proven to have linear to super-linear convergence [XX]. A better rate of convergence is observed using Newton-Raphson, which uses second-order information in the form of a Hessian to reach a critical point. The Hessian represents a linearized sensitivity model of the gradient, $\nabla f_i$. During the consensus step in distributed optimization, second-order information of each sub-problem would provide a better weighting scheme and improve the convergence rate compared to simple fixed-point iteration. However, communicating a Hessian for each iteration of the Gauss-Seidel is computationally intractable for larger optimization problems. We propose an aggregate sensitivity model, that is derived from a Thevenin impedance model from circuits. 

%Consider a single sub-problem state-space equation from \eqref{eq:gd_flow_inductor_subproblem}. This model is characterized by an independent state-vector $\x_i$ and an inductor current, $I_i^L$ which couples sub-system with the central agent. To improve the convergence of the central agent step, we can improve on the model by extracting the linear sensitivity of each sub-problem.  
%The successive chord model is a commonly used method that can perform similarly to the Newton-Raphson method within a normal operational range. To further improve of ECADO, 
We can then integrate the constant linear sensitivity model into the central agent state-space equations as follows:\begin{align}
    Z_c\dot{\x_c}^{k+1}(t) + \sum_{i=1}^n I_i^{L^{k+1}} =0 \\
    L\dot{I}_i^{L^{k+1}} = \x_c^{k+1} - ( I_i^{L^{k+1}} \bar{R}_i^{th} + \x_i^{k+1} - I_i^{L^k} \bar{R}_i^{th})\label{eq:central_chord_ode}.
\end{align}
%The addition of $\bar{R}_i^{th}$ weights the update of the inductor currents in \eqref{eq:central_chord_ode} using the averaged linear sensitivity, which are then used to drive the dynamics of the central agent state vector, $\x_c$. 

\subsubsection{Backward-Euler Integration}

 To simulate the response of the central agent states, we solve for the state at each time-point in \eqref{eq:central_chord_ode} according to:
 \begin{align}
     \x_c^{k+1}(t+\Delta t) = \x_c^{k+1}(t) - \int_{t}^{t+\Delta t}Z_c^{-1}\sum_{i=1}^{n}I_i^{L^{k+1}}(t)dt \\
     I_i^{L^{k+1}}(t+\Delta t) = I_i^{L^{k+1}}(t) + \frac{1}{L}\int_{t}^{t+\Delta t} \left( \x_c^{k+1}(t) - ( I_i^{L^{k+1}}(t) \bar{R}_i^{th} + \x_i^{k+1} - I_i^{L^{k}}\bar{R}_i^{th} \right ) dt, \label{eq:central_ode_int}
 \end{align}
where, $\x_i^{k+1}(t)$, is a known value from the G-S iteration (from each sub-problem).

Note that the ODE in \eqref{eq:central_ode_int} is linear. We approximate the integral on the right hand side using a Backward-Euler (B.E.) numerical integration since it provides two primary advantages over other numerical integration methods:
\begin{enumerate}
    \item B.E. is numerically stable, thereby it avoids issues of divergence
    \item It adds positive values to only the diagonal entries of the G-S iteration matrix, which improves the numerical conditioning and convergence rate.
\end{enumerate}

The B.E. step approximates the states at time-point, $t+\Delta t$, of \eqref{eq:central_ode_int} by:
\begin{align}
         \x_c^{k+1}(t+\Delta t) = \x_c^{k+1}(t) - \Delta t\left(Z_c^{-1}\sum_{i=1}^{n}I_i^{L^{k+1}}(t+\Delta t) \right) \\
     I_i^{L^{k+1}}(t+\Delta t) = I_i^{L^{k+1}}(t) + \frac{\Delta t}{L}  \left( \x_c^{k+1}(t+\Delta t) - ( I_i^{L^{k+1}}(t+\Delta t) \bar{R}_i^{th} + \x_{i}^{k+1}(t+\Delta t) - I_i^{L^k}(t+\Delta t)\bar{R}_i^{th}) \right ). \label{eq:central_ode_be_step}
\end{align}

This creates the following set of linear equations:

\begin{equation}
    \begin{bmatrix}
        1+\frac{\Delta t \bar{R}_1^{th}}{L} & 0 & \ldots & -\frac{\Delta t}{L} \\
        0 & 1+\frac{\Delta t \bar{R}_2^{th}}{L} & \ldots & -\frac{\Delta t}{L}\\
        0 & 0 & \ddots & \frac{-\Delta t}{L} \\
        -\Delta t Z_c^{-1} & -\Delta t Z_c^{-1} & \ldots & 1
    \end{bmatrix} \begin{bmatrix}
        I_1^{L^{k+1}}(t+\Delta t) \\ 
        I_2^{L^{k+1}}(t+\Delta t) \\
        \vdots \\
        \x_c^{k+1}(t+\Delta t)
    \end{bmatrix}
    = \frac{\Delta t}{L}\begin{bmatrix}
       -\x_1^{k+1} + I_1^{k}(t)\bar{R}_1^{th} \\
       -\x_2^{k+1} + I_2^{k}(t)\bar{R}_2^{th} \\
        \vdots \\
        0
    \end{bmatrix} \label{eq:central_agent_be_step2}
\end{equation}

with a time-step of $\Delta t$.

For a fixed $\Delta t$, the left-hand side in \eqref{eq:central_agent_be_step2} is a constant, invertible matrix. To reduce the runtime cost of evaluating the central agent states in \eqref{eq:central_agent_be_step2}, we LU factor the matrix once for the fixed chord model (constant aggregate sensitivity) and fixed time-step, $\Delta t$. This effectively reduces computing the central agent states to a simple forward-backward substitution.

One of the main benefits of ECADO's Backward-Euler integration is that it is an \emph{implicit} numerical integration method. Optimization literature has generally used \emph{explicit numerical integration} methods, which use previous known values of $\x_c$ to approximate the right-hand side integral of \eqref{eq:central_ode_int}. However, explicit numerical integration, as well as gradient descent, are prone to numerical instability \cite{agarwal2023equivalent}, and must be managed by highly restrictive time-steps to avoid divergence or oscillations.  However, the implicit Backward-Euler integration method is proven to be numerically stable \cite{lawrence1995electronic}, which allows ECADO to utilize large time-steps that will not diverge or numerically oscillate for stable systems.

\subsection{Convergence of EC Consensus Algorithm}

ECADO's consensus algorithm  \eqref{eq:central_agent_be_step2} introduces three new concepts to distributed optimization that facilitate faster convergence:\\
1) A PI controller using flow variables to connect the central agent to each sub-problem \\
2) A weighting scheme for calculating the flow that is based on an averaged sensitivity model, $\bar{R}_i^{th}$ \\
3) A Backward-Euler integration that allows the central agent to take larger step-sizes thereby reducing the number of iterations required to reach a critical point 

 We analyze the effect of these new methods on the convergence rate of the proposed consensus algorithm in \eqref{eq:central_agent_be_step2}.

\subsubsection{Convergence Guarantees}

To determine the convergence rate, we study ECADO's differential equations in the following matrix form:
\begin{equation}
\begin{bmatrix}
    Z_1(\x_1) &0 & \ldots & 0 & 0 & \ldots & 0\\
    0 & Z_2(\x_2) &\ldots & 0 & 0 &\ldots &0\\
    0 & 0 & \ddots & 0 & 0 &\ldots & 0 \\
    0 & 0 & \ldots & L & 0 & \ldots & 0\\
    0 &0 & \ldots & 0 & L & \ldots & 0\\
    0 &0 & \ldots & 0 & 0 & \ddots & 0\\
    0 &0 & \ldots & 0 & 0 & \ldots & Z_c
    \end{bmatrix} \begin{bmatrix}
        \dot{\x}_1\\
        \dot{\x}_2\\
        \vdots \\
        \dot{I}_1^L \\
        \dot{I}_2^L \\
        \vdots \\
        \dot{\x}_c
    \end{bmatrix} = 
    \begin{bmatrix}
        0 & 0 & \ldots & 1 & 0 & \ldots & 0\\
        0 & 0 & \ldots & 0 & 1 &\ldots & 0\\
        0 & 0 & \ddots &0 & 0 & \ddots & 0 \\
        -1 & 0 & \ldots & -\bar{R}_1^{th} & 0 & \ldots &1\\
        0 & -1 & \ldots & 0 & -\bar{R}_2^{th} & \ldots & 1 \\
        0 & 0 & \ddots & 0 & 0 &\ddots & 1\\
        0 & 0 & \ldots &-1 & -1 & \ldots &0
    \end{bmatrix}
    \begin{bmatrix}
        \x_1 \\
        \x_2 \\
        \vdots \\
        I_1^L \\
        I_2^L \\
        \vdots \\
        \x_c
    \end{bmatrix} + \begin{bmatrix}
        -\nabla f_1(\x_1) \\
        -\nabla f_1(\x_2) \\ 
        \vdots \\
        I_1^{L^k} \bar{R}_1^{th}\\
        I_2^{L^k} \bar{R}_2^{th} \\
        \vdots \\
        0
    \end{bmatrix} \label{eq:ec_ode_matrix}
\end{equation}

\begin{theorem}
    The Gauss-Seidel process \eqref{eq:ec_ode_matrix} is guaranteed to converge using an integration method with $\Delta t \rightarrow 0$.
\end{theorem}
\begin{proof}
    A sufficient condition for the convergence of any convex or nonconvex optimization function using ECADO is that the left-hand side, block-diagonal $Z-$matrix is diagonally dominant. Convergence is guaranteed to be achievable since there exists a $\Delta t$ of sufficiently small magnitude for which each block diagonal in the $Z-$matrix is diagonally-dominant. The diagonal entries of $Z_{ii}$ satisfy the assumptions in \eqref{a5} and ensure entries on the diagonal are greater than 0 when we design the controller gain to be $L>0$.  A more complete proof of convergence for the diagonal dominance is provided in Appendix \ref{sec:convergence_proof_appendix}. 
\end{proof} 

\subsubsection{Convergence Rate of Distributed EC}
We analyze the convergence rate of ECADO for a quadratic objective function, $f_i(\x_i)$:
%To analyze the convergence rate of the distributed EC approach, we restrict each objective $f_i(\x_i)$ to be a quadratic function of the form:
\begin{equation}
    f_i(\x_i) = \frac{1}{2} A_i \|\x_i\|^2 + B_i \x_i + C_i, \label{eq:quadratic_obj}
\end{equation}
where $A_i \succ 0$.
ECADO's state-space matrix equation  \eqref{eq:ec_ode_matrix} for solving the quadratic objective function \eqref{eq:quadratic_obj} is:
\begin{multline}
\label{eq:quadratic_ode}
    \begin{bmatrix}
        \dot{\x}_1\\
        \dot{\x}_2\\
        \vdots \\
        \dot{I}_1^L \\
        \dot{I}_2^L \\
        \vdots \\
        \dot{\x}_c
    \end{bmatrix} = \\ \begin{bmatrix}
    Z_1(\x_1) &0 & \ldots & 0 & 0 & \ldots & 0\\
    0 & Z_2(\x_2) &\ldots & 0 & 0 &\ldots &0\\
    0 & 0 & \ddots & 0 & 0 &\ldots & 0 \\
    0 & 0 & \ldots & L & 0 & \ldots & 0\\
    0 &0 & \ldots & 0 & L & \ldots & 0\\
    0 &0 & \ldots & 0 & 0 & \ddots & 0\\
    0 &0 & \ldots & 0 & 0 & \ldots & Z_c
    \end{bmatrix}^{-1} \left (  \begin{bmatrix}
        -A_1 & 0 & \ldots & 1 & 0 & \ldots & 0\\
        0 & -A_2 & \ldots & 0 & 1 &\ldots & 0\\
        0 & 0 & \ddots &0 & 0 & \ddots & 0 \\
        -1 & 0 & \ldots & -\bar{R}_1^{th} & 0 & \ldots &1\\
        0 & -1 & \ldots & 0 & -\bar{R}_2^{th} & \ldots & 1 \\
        0 & 0 & \ddots & 0 & 0 &\ddots & 1\\
        0 & 0 & \ldots &-1 & -1 & \ldots &0
    \end{bmatrix}
    \begin{bmatrix}
        \x_1 \\
        \x_2 \\
        \vdots \\
        I_1^L \\
        I_2^L \\
        \vdots \\
        \x_c
    \end{bmatrix} + \begin{bmatrix}
        -B_1 \\
        -B_2 \\ 
        \vdots \\
        I_1^{L^k} \bar{R}_1^{th}\\
        I_2^{L^k} \bar{R}_2^{th} \\
        \vdots \\
        0
    \end{bmatrix}   \right )
\end{multline}

\begin{multline}
    \implies \begin{bmatrix}
        \dot{\x}_1(t)\\
        \dot{\x}_2(t)\\
        \vdots \\
        \dot{I}_1^L(t) \\
        \dot{I}_2^L(t) \\
        \vdots \\
        \dot{\x}_c(t)
    \end{bmatrix} = \begin{bmatrix}
        -Z_1^{-1}(\x_1)A_1 & 0 & \ldots & Z_1^{-1}(\x_1) & 0 & \ldots & 0\\
        0 & -Z_2^{-1}(\x_2)A_2 & \ldots & 0 & Z_2^{-1}(\x_2) &\ldots & 0\\
        0 & 0 & \ddots &0 & 0 & \ddots & 0 \\
        -L^{-1} & 0 & \ldots & -L^{-1}\bar{R}_1^{th} & 0 & \ldots &L^{-1}\\
        0 & -L^{-1} & \ldots & 0 & -L^{-1}\bar{R}_2^{th} & \ldots & L^{-1} \\
        0 & 0 & \ddots & 0 & 0 &\ddots & L^{-1}\\
        0 & 0 & \ldots &-Z_c^{-1} & -Z_c^{-1} & \ldots &0
    \end{bmatrix}
    \begin{bmatrix}
        \x_1(t) \\
        \x_2(t) \\
        \vdots \\
        I_1^L(t) \\
        I_2^L(t) \\
        \vdots \\
        \x_c(t)
    \end{bmatrix} + \begin{bmatrix}
        -Z_c^{-1}B_1 \\
        -Z_c^{-1}B_2 \\ 
        \vdots \\
        L^{-1}I_1^{L}(t-\Delta t) \bar{R}_1^{th}\\
        L^{-1}I_2^{L}(t-\Delta t) \bar{R}_2^{th} \\
        \vdots \\
        0
    \end{bmatrix}
\end{multline}

The right-hand side matrix can be abstracted as $A_{ec}$ and separated into a block-diagonal ($D_{ec}$), block lower-triangular ($L_{ec}$) and block upper-triangular ($U_{ec}$) matrices as:
\begin{equation}
    A_{ec}=D_{ec} +L_{ec} + U_{ec}.
\end{equation}

ECADO solves the ODE using a Backward-Euler integration and a G-S iteration according to \cite{alberto_numerical_methods}:
\begin{equation}
    \begin{bmatrix}
        \x_1^{k+1} \\
        \x_2^{k+1} \\
        \vdots \\
        I_1^{L^{k+1}} \\
        I_2^{L^{k+1}} \\
        \vdots \\
        \x_c^{k+1}
    \end{bmatrix} = (I-\Delta t(D_{ec}+L_{ec}))^{-1}(I + \Delta t U_{ec})    \begin{bmatrix}
        \x_1^k \\
        \x_2^k \\
        \vdots \\
        I_1^{L^k} \\
        I_2^{L^k} \\
        \vdots \\
        \x_c^k
    \end{bmatrix} + \Delta t \begin{bmatrix}
        -Z_c^{-1}B_1 \\
        -Z_c^{-1}B_2 \\ 
        \vdots \\
        L^{-1}I_1^{L}(t-\Delta t) \bar{R}_1^{th}\\
        L^{-1}I_2^{L}(t-\Delta t) \bar{R}_2^{th} \\
        \vdots \\
        0
    \end{bmatrix}.
\end{equation}

Taking the difference between the $k+1$ iteration and $k$ iteration:
\begin{equation}
    \begin{bmatrix}
    \x_1^{k+1} \\
    \x_2^{k+1} \\
    \vdots \\
    I_1^{L^{k+1}} \\
    I_2^{L^{k+1}} \\
    \vdots \\
    \x_c^{k+1}
\end{bmatrix} - \begin{bmatrix}
    \x_1^{k} \\
    \x_2^{k} \\
    \vdots \\
    I_1^{L^{k}} \\
    I_2^{L^{k}} \\
    \vdots \\
    \x_c^{k+1}
\end{bmatrix}= (I-\Delta t(D_{ec}+L_{ec}))^{-1}(I + \Delta t U_{ec}) \left (   \begin{bmatrix}
    \x_1^k \\
    \x_2^k \\
    \vdots \\
    I_1^{L^k} \\
    I_2^{L^k} \\
    \vdots \\
    \x_c^k
\end{bmatrix}  - \begin{bmatrix}
    \x_1^{k-1} \\
    \x_2^{k-1} \\
    \vdots \\
    I_1^{L^{k-1}} \\
    I_2^{L^{k-1}} \\
    \vdots \\
    \x_c^{k-1}
\end{bmatrix} \right ).
\end{equation}
we define the iterative matrix for the G-S process, $G_{ec}$, as:
\begin{equation}
    G_{ec} = (I-\Delta t(D_{ec} + L_{ec}))^{-1}(I + \Delta t U_{ec}).
\end{equation}
The convergence rate of G-S is equal to:
\begin{equation}
O(e^{\rho(G_{ec})k})
\end{equation}
where $\rho$
is the spectral radius of $G_{ec}$. To study the convergence rate of ECADO, we analyze the spectral radius of $G_{ec}$, which is upper bounded by:
\begin{equation}
    \rho(G_{ec}) \leq \rho ((I-\Delta t(D_{ec} + L_{ec}))^{-1}) \rho(I + \Delta t U_{ec}).
\end{equation}

The eigenvalues for the block upper-triangular matrix $(I+\Delta tU_{ec})$ are equal the eigenvalues of the diagonal block entries, $I$. Therefore with no dependence on $\Delta t$, hence no control of convergence using step-size, the spectral radius of $\rho(I+\Delta t U_{ec})$ is 1. 

The spectral radius of $I-\Delta t(D_{ec} + L_{ec})^{-1}$ is:
\begin{equation}
    \rho((I-\Delta t(D_{ec} + L_{ec}))^{-1}) = \max \left [ \frac{1}{eig(I+\Delta tZ_i^{-1}(\x_i)A_i)}, \frac{1}{eig(1+L^{-1}R_i^{th})}, eig(\Delta t Z_c^{-1}) \right ]
\end{equation}

It follows that the spectral radius of $G_{ec}$ is bounded by:
\begin{equation}
    \rho(G_{ec}) \leq \max \left [ \frac{1}{eig(I+\Delta tZ_i^{-1}(\x_i)A_i)}, \frac{1}{eig(1+L^{-1}R_i^{th})}, eig(\Delta t Z_c^{-1}) \right ], \label{eq:ec_spectral_radius}
\end{equation}

and the upperbound on the convergence rate is:
\begin{equation}
    O(e^{\rho(G_{ec})k}) = O(e^{\max \left [ \frac{1}{eig(I+\Delta tZ_i^{-1}(\x_i)A_i)}, \frac{1}{eig(1+\Delta t L^{-1}R_i^{th})}, eig(\Delta t Z_c^{-1}) \right ]k}).
\end{equation}

The convergence rate of the G-S iteration is influenced by the diagonal dominance of $D_{ec}$. Generally, larger values on the diagonal provide faster convergence to the steady-state \cite{alberto_numerical_methods}. The features of ECADO (adding flow variables, sensitivity model, and using Backward-Euler) add only positive elements (namely $Z_c, L$ and $\Delta t$ ) to the diagonal, thereby only improving the convergence rate. For example, we can pre-select a value of $\Delta t L$ that decreases  $\frac{1}{eig(1+\Delta t L^{-1}R_i^{th})}$. Additionally, a larger $Z_c$ will decrease $eig(\Delta t Z_c^{-1})$, and provide faster convergence.

\paragraph{Effect of Backward-Euler Integration}

The implicit Backward-Euler integration provides better convergence by adding a $\Delta t Z_c^{-1}$ to the diagonal of the $Z-$matrix, which only increases diagonal dominance of the G-S matrix unlike other numerical integration techniques. For example, Forward-Euler integration of \eqref{eq:central_ode_int}  would only add $I$ on the diagonal.

\paragraph{Effect of Linear Sensitivity Model}

 The aggregate sensitivity, $\bar{R}_i^{th}$, incorporates the sensitivity values into the block diagonal, $D_{ec}$, thereby improving the convergence rate of Gauss-Seidel.  

To analyze the effect of $\bar{R}_i^{th}$, the limit as  $\bar{R}_i^{th} \rightarrow 0$ (indicating we remove the linear sensitivity model) shifts the upper bound of the spectral radius to:
\begin{equation}
    \rho(G_{ec})_{\bar{R}_{th}\rightarrow0} \leq \max \left [ \frac{1}{eig(I+\Delta tZ_i^{-1}(\x_i)A_i)}, \frac{1}{eig(I)}, eig(\Delta t Z_c^{-1}) \right ], 
\end{equation}
which is bounded by a spectral radius of 1, which would produce a slower rate of convergence.

\paragraph{Convergence Rate Comparison} 
We can tune the values $L,Z_c$ and $\Delta t$ to ensure a convergence rate of ECADO, for the quadratic objective function, is comparable to or exceeds the convergence rate of state-of the art distributed optimization methods. For this, we compare with centralized gradient-descent (CGD) \cite{centralized_gradient_descent}, ADMM \cite{admm_boyd2011distributed} and DANE \cite{dane_shamir2014communication}. The error bound after $k$ iterations for each method using G-S are compared in Appendix \ref{sec:convergence_proof_appendix}. The errors are a function of the spectral radius of the iterative matrix given by:
\begin{equation}
    \|\x^k - \x^* \|^2 = O(e^{\rho(G)k}).
\end{equation}
To compare the convergence rates of the methods, we analyze spectral radius for the respective iterative matrix in Table \ref{table:convergence_rate_comparison}.

\begin{table}[h]
\centering 
\label{table:convergence_rate_comparison}
\begin{tabularx}{\textwidth} { 
  | >{\raggedright\arraybackslash}X 
  | >{\centering\arraybackslash}X 
  | >{\centering\arraybackslash}X | >{\centering\arraybackslash}X |>{\centering\arraybackslash}X |}
 \hline
 - & One-Shot & ADMM & DANE & ECADO \\
 \hline
  $\rho(G)$  & $\max (eig(1-\frac{\alpha}{n}\sum_{i=1}^n A_i))$  & $\max eig(1-\alpha(A_i + \eta I))$ & $\max eig(1-\frac{\alpha}{n}\sum_{i=1}^n (A_i + \mu I)^{-1}A_i)$ &  $ \max \left [  \rho_1, \rho_2, \rho_3 \right ]$  \\
 %$\rho(G)$  & $\max (eig(1-\frac{\alpha}{n}\sum_{i=1}^n A_i))$  & $\max eig(1-\alpha(A_i + \eta I))$ & $\max eig(1-\frac{\alpha}{n}\sum_{i=1}^n (A_i + \mu I)^{-1}A_i)$ &  $ \max \left [  \frac{1}{eig(I+\Delta tZ_i^{-1}(\x_i)A_i)} , \frac{1}{eig(1+L^{-1}R_i^{th})}, eig(\Delta t Z_c^{-1}) \right ]$  \\
\hline
\end{tabularx}
\end{table}

The spectral radius of the iterative matrix of ECADO, $G_{ec}$ is bounded by three values:
\begin{align}
    \rho(G_{ec}) \leq \max(\rho_1, \rho_2, \rho1) \\
    \rho_1 = \max \frac{1}{eig(I+\Delta tZ_i^{-1}(\x_i)A_i)} \label{eq:rho1}\\
    \rho_2 = \max  \frac{1}{eig(1+L^{-1}R_i^{th})} \label{eq:rho2}\\
    \rho_3 = \max eig(\Delta t Z_c^{-1})  \label{eq:rho3}
\end{align}

We can select the three hyperparameters, $Z_c$, $L$, and $\Delta t$, to optimally tune the three spectral radiuses ($\rho_1, \rho_2, \rho_3$). 

 By increasing $Z_c \rightarrow \inf$, we effectively decrease $\rho_3 \rightarrow 0$ \eqref{eq:rho3} and improve the upper bound on the convergence rate.

We can also reduce the spectral radius of $\rho_2$ \eqref{eq:rho2} by setting $L_i = \beta R_i^{th}$, which sets $\rho_2$ to:
\begin{equation}
    \rho_2 = \frac{1}{1+\beta}.
\end{equation}
A large value of $\beta>0$ ensures that the spectral radius is less than 1.

It should be noted that $\Delta t$ can be reduced to decrease $\rho_3$. However, $\Delta t$ must also satisfy conditions for numerical accuracy. We next discuss an adaptive method to select a time-step, $\Delta t$.

\subsection{Adaptive Time-Step Selection}

ECADO uses a B.E. integration to approximate the continuous time integration on the right-hand side of \eqref{eq:central_ode_int}. %The accuracy of the approximation is a function of the type of numerical integration method as well as the time-step, $\Delta t$. Without proper consideration of accuracy, the numerical integration may diverge from the intended continuous-time ODE. %The local error introduced at each time-step by the approximation of the numerical integration may accumulate to cause divergence from the steady-state as well as numerical oscillations. These issues are demonstrates in [XX].
 We devise an adaptive time-step selection algorithm that ensures the iterations reach  a steady state by satisfying two conditions: local accuracy and G-S convergence.

\paragraph{Numerical Accuracy Condition: }
 The local truncation error for the B.E> approximation in \eqref{eq:central_ode_int} can be estimated by \cite{lawrence1995electronic}:
\begin{equation}
    \varepsilon_{BE}^C = -\frac{-\Delta t}{2Z_c}\left [ \sum_{i=1}^{n}I_i^{L^{k+1}}(t) - \sum_{i=1}^{n}I_i^{L^{k+1}}(t+\Delta t)  \right ], \label{eq:be_lte_cap}
\end{equation}
for the central agent state ODE. The local truncation error for the flow variable in \eqref{eq:central_chord_ode_inductor} can be estimated similarly \cite{lawrence1995electronic} as:
\begin{multline}
    \varepsilon_{BE_i}^L = -\frac{\Delta t}{2L}[ (\x_c^{k+1}(t) - I_i^{L^{k+1}}(t)\bar{R}_i^{th} + \x_i^{k+1}(t) - I_i^{L^k}(t)\bar{R}_i^{th}) \\ - ((\x_c^{k+1}(t+\Delta t) - I_i^{L^{k+1}}(t+\Delta t)\bar{R}_i^{th} + \x_i^{k+1}(t+\Delta t) - I_i^{L^k}(t+\Delta t)\bar{R}_i^{th}))  ]. \label{eq:be_lte_ind}
\end{multline}

Using these measures, we introduce the following accuracy criteria:
\begin{equation}
    max(|\varepsilon_{BE}|) \leq \delta,
\end{equation}
where $\delta$ is a predefined tolerance for the maximum local truncation error. This ensures that the local approximation of Backward-Euler does not diverge from the continuous time trajectory.

\paragraph*{Gauss-Seidel Convergence Condition: }

Although Backward-Euler integration is numerically stable and guaranteed to converge to the steady-state for any given $\Delta t$ \cite{lawrence1995electronic}, the choice of $\Delta t$ influences the convergence of the inner G-S process, which is defined as:
\begin{equation}
    \|\dot{\x_c}^{k+1} - \dot{\x_c}^k\|^2 \leq \|\dot{\x_c}^k - \dot{\x_c}^{k-1}\|^2. \label{eq:gs_criteria}
\end{equation}
The proof of this convergence criteria lies in Appendix \ref{sec:convergence_proof_appendix}, with $j+1 = k$. 

At each iteration, we must select a step-size that ensures GS is converging toward the true solution according to the condition in \eqref{eq:gs_criteria}. However, directly evaluating $\dot{\x}_c$ would require additional numerical approximations. Instead, we recognize that the criteria in \eqref{eq:gs_criteria} is equivalent to:
\begin{equation}
    \| \sum_{i=1}^{n}I_i^{L^{k+1}} - \sum_{i=1}^{n}I_i^{L^{k}} \| \leq \| \sum_{i=1}^{n}I_i^{L^{k}} - \sum_{i=1}^{n}I_i^{L^{k-1}} \|, \label{eq:gs_criteria2}
\end{equation} 
since  $\dot{\x_c} = Z_c^{-1} \sum_{i=1}^{n} I_i^L$. The condition in \eqref{eq:gs_criteria2} is easily evaluated at each iteration.

To select a $\Delta t$ that satisfies the accuracy and convergence conditions, we use a backtracking line-search method at each iteration of the central agent. In this approach, shown in Algorithm \ref{adaptive-time-step}, the time-step (i.e., step-size) is adaptively reduced by a factor $\eta\in(0,1)$ until the two conditions are satisfied.

Importantly, following the convergence proof in Appendix \ref{sec:convergence_proof_appendix}, there exists a $\Delta t>0$ that guarantees G-S convergence. Additionally, by the proof in \cite{lawrence1995electronic}, there also exists a $\Delta t>0$ that will satisfy the accuracy criteria. This indicates that Algorithm \ref{adaptive-time-step} is bounded. In practice, Algorithm \ref{adaptive-time-step} rarely requires multiple iterations to find a $\Delta t$ that satisfies the two conditions.

\begin{algorithm}
    \caption{Adaptive Time-Stepping Method}
    \label{adaptive-time-step}
    \textbf{Input: } $\eta\in(0,1)$
    \begin{algorithmic}[1]
    
    \STATE{\textbf{do while} $\| \sum_{i=1}^{n}I_i^{L^{k+1}} - \sum_{i=1}^{n}I_i^{L^{k}} \| \leq \| \sum_{i=1}^{n}I_i^{L^{k}} - \sum_{i=1}^{n}I_i^{L^{k-1}} \|$ and $max(|\varepsilon_{BE}|) \leq \delta$}
    %\STATE{\hspace*{\algorithmicindent}$t\gets 0$}
    \STATE{\hspace*{\algorithmicindent}$\Delta t = \eta \Delta t$}
    
    \STATE{\hspace*{\algorithmicindent}Solve for $\x_c^{k+1}(t), I_i^{L^{k+1}}(t)$: }
    \STATE{\hspace*{\algorithmicindent}\hspace*{\algorithmicindent}\hspace*{\algorithmicindent} $    \begin{bmatrix}
        1-\frac{\Delta t \bar{R}_1^{th}}{L} & 0 & \ldots & -\frac{\Delta t}{L} \\
        0 & 1-\frac{\Delta t \bar{R}_2^{th}}{L} & \ldots & -\frac{\Delta t}{L}\\
        0 & 0 & \ddots & \frac{-\Delta t}{L} \\
        -\Delta t Z_c^{-1} & -\Delta t Z_c^{-1} & \ldots & 1
    \end{bmatrix} \begin{bmatrix}
        I_1^{L^{k+1}}(t+\Delta t) \\ 
        I_2^{L^{k+1}}(t+\Delta t) \\
        \vdots \\
        \x_c^{k+1}(t+\Delta t)
    \end{bmatrix}
    = \frac{\Delta t}{L}\begin{bmatrix}
        -\x_1^{k+1} + I_1^{k}(t)\bar{R}_1^{th} \\
        -\x_2^{k+1} + I_2^{k}(t)\bar{R}_2^{th} \\
         \vdots \\
         0
     \end{bmatrix}$}
    \RETURN $\Delta t$
    \end{algorithmic}
    \end{algorithm}

\subsection{Full Distribted EC Algorithm}

The complete algorithm for solving the distributed EC model is shown in Algorithm \ref{full-gauss-seidel}. The algorithm begins by precomputing the average sensitivity models, $\bar{R}_i^{th}$, in line 6 and the LU factor of \eqref{eq:central_agent_be_step2} to reduce the optimization runtime. The impacts of the B.E. numerical integration and  sensitivity models are captured by the state variables in line 14.

\begin{algorithm}
    \caption{Modified Gauss-Seidel for Distributed EC Model}
    \label{full-gauss-seidel}
    \textbf{Input: } $\nabla f(\cdot)$,$\x(0)$, $\Delta T>0$
    \begin{algorithmic}[1]
    \STATE{$\x_c \gets \x(0)$}
    \STATE{$\x_i \gets \x(0)$}
    \STATE{$I_i^L \gets 0$}
    \STATE{$t_1 \gets 0$}
    \STATE{$t_2 \gets \Delta T$}

    \STATE{Precompute $\bar{R}_i^{th} \; \forall i\in[1,n]$} 
    \STATE{Precompute LU-factor for matrix in \eqref{eq:central_agent_be_step2}}
    
    \STATE{\textbf{do while} $\|\dot{\x}_c\|^2 > 0$}
    %\STATE{\hspace*{\algorithmicindent}$t\gets 0$}
    \STATE{\hspace*{\algorithmicindent}$\x_c^k \gets \x_c^{k+1}$}
    \STATE{\hspace*{\algorithmicindent}$\x_i^{k} \gets \x^{i^{k+1}}$}
    \STATE{\hspace*{\algorithmicindent}\textbf{do while} $\|\x_c^{k+1} - \x_c^{k}\|^2 > 0$}
    \STATE{\hspace*{\algorithmicindent}\hspace*{\algorithmicindent}Parallel Solve for $\x_i^{k+1}(t)$: $Z(\x_i^{k+1})\dot{\x}_i^{k+1}(t) = -\nabla f(\x_i^{k+1}(t)) + I_i^{L^k}(t) \;\;\; \forall i \in [1,n]\;\;\; \forall t\in[t_1,t_2] $}

    \STATE{\hspace*{\algorithmicindent}\hspace*{\algorithmicindent}Solve for $\x_c^{k+1}(t), I_i^{L^{k+1}}(t)$: }
    \STATE{\hspace*{\algorithmicindent}\hspace*{\algorithmicindent}\hspace*{\algorithmicindent} Select $\Delta t$ according to Algorithm \ref{adaptive-time-step}}
    \STATE{\hspace*{\algorithmicindent}\hspace*{\algorithmicindent}\hspace*{\algorithmicindent} $    \begin{bmatrix}
        1-\frac{\Delta t \bar{R}_1^{th}}{L} & 0 & \ldots & -\frac{\Delta t}{L} \\
        0 & 1-\frac{\Delta t \bar{R}_2^{th}}{L} & \ldots & -\frac{\Delta t}{L}\\
        0 & 0 & \ddots & \frac{-\Delta t}{L} \\
        -\Delta t Z_c^{-1} & -\Delta t Z_c^{-1} & \ldots & 1
    \end{bmatrix} \begin{bmatrix}
        I_1^{L^{k+1}}(t+\Delta t) \\ 
        I_2^{L^{k+1}}(t+\Delta t) \\
        \vdots \\
        \x_c^{k+1}(t+\Delta t)
    \end{bmatrix}
    = \frac{\Delta t}{L}\begin{bmatrix}
        -\x_1^{k+1} + I_1^{k}(t)\bar{R}_1^{th} \\
        -\x_2^{k+1} + I_2^{k}(t)\bar{R}_2^{th} \\
         \vdots \\
         0
     \end{bmatrix}$}
    \STATE{\hspace*{\algorithmicindent} $t_1 += \Delta T$}
    \STATE{\hspace*{\algorithmicindent} $t_2 += \Delta T$}
    \RETURN $\x_c$
    \end{algorithmic}
    \end{algorithm}

\section{Results}
We evaluate the performance of ECADO by solving convex and nonconvex distributed optimization problems in the fields of machine learning and power systems. We compare our approach to state-of-the-art centralized distributed optimization algorithms that include centralized gradient descent (CGD) \cite{centralized_gradient_descent}, ADMM \cite{admm_boyd2011distributed} and DANE \cite{dane_shamir2014communication}, ECADO uses the weighting scheme provided by the sensitivity model, $\bar{R}_i^{th}$, and applies the adaptive B.E. scheme for faster convergence to the critical point. %The distributed EC algorithm follows the methods described in section XX and use a weighting scheme provided by a pre-computed average sensitivity model, $R_i^{th}$ and an adaptive Backward-Integration method to select the time-step $\Delta t$. 

%We demonstrate the efficacy of our approach on both convex and nonconvex applications. The convex application is optimizing binary classification of synthetic data and a real dataset using logistic regression. The nonconvex applications are training a neural network on distributed dataset, and a large security-constrained optimal power flow (SCOPF) optimization.

\iffalse
The hyperparameters for each optimizer are detailed in Table XX and are provided by the examples in [XX]. Each optimizer performs a single inner-iteration of each sub-problem (i.e., $t_2 - t_1 = \Delta t$).

\begin{table}
    \centering
    \begin{tabularx}{\textwidth} { 
        | >{\raggedright\arraybackslash}X 
        | >{\centering\arraybackslash}X 
        | >{\centering\arraybackslash}X | >{\centering\arraybackslash}X |>{\centering\arraybackslash}X |}
       \hline
       - & One-Shot & ADMM & DANE & ECCO \\
       \hline
    \end{tabularx}
\end{table}
\fi
For ADMM, DANE, and central gradient-descent (CGD), we a use symmetric fastest
distributed linear averaging (FDLA) matrices \cite{xiao2004fast} for aggregating $\x_c$, and a constant step-size of $1e-4$. The machine-learning experiments and parameters are based on the code in \cite{yuejie_li2020communication}.

\subsection{Logistic Regression}
We used a regularized logistic regression to solve a binary classification of a synthetic \cite{yuejie_li2020communication} and a gisette dataset \cite{gisetteData}. The training datasets are split amongst $m=20$ agents along with a centralized agent to provide the consensus step. Each agent is optimizing the following loss function:
\begin{equation}
    f_i(\x) = \frac{1}{l} \sum_{j=1}^{l}\left [ b_i^{j} (\frac{1}{1+exp(\x^T a_i^j)}) (1-b_i^j)log(\frac{exp(\x^T a_i^j)}{ 1+exp(\x^T a_i^j)})  \right ] + \frac{\lambda}{2} \|\x\|^2 \label{eq:logistic_regression_loss}
\end{equation}
where $l$ is the size of each training set, $a_i^j,b_i^j\in{0,1}$ are samples from the training set for agent $i$ and $\lambda$ is a regularizer term.

We train each optimizer for the synthetic dataset provided in \cite{yuejie_li2020communication} as well as the gisette dataset \cite{gisetteData} to convergence (defined as $f(\x)=1-^{-10}$ for the synthetic dataset and $f(\x)=10^{-4}$ for the gisette dataset. Each agent receives $l=300$ training samples of dimension $d=5000$. 

The convergence of each optimizer is shown in Figure \ref{fig:logistic_regression_synthetic} for training on the synthetic dataset over the number of outer iterations and in Figure \ref{fig:logistic_regression_gisette} for training on the gisette dataset. ECADO demonstrated significantly improved convergence compared to the other optimizers. In this example, all sensitivity models, $R_i^{th}$, are identical, thereby reducing the effect of including an average sensitivity model. This corresponds to a sensitivity-based weighting scheme for which all agents have an identical weight. The improved convergence, therefore, would be attributed to the Backward-Euler approximation and the integration-based controller. 

\begin{figure}
    \centering
    \begin{minipage}[t]{.45\linewidth}
      \centering
      \label{fig:logistic_regression_synthetic}\includegraphics[width=0.8\columnwidth]{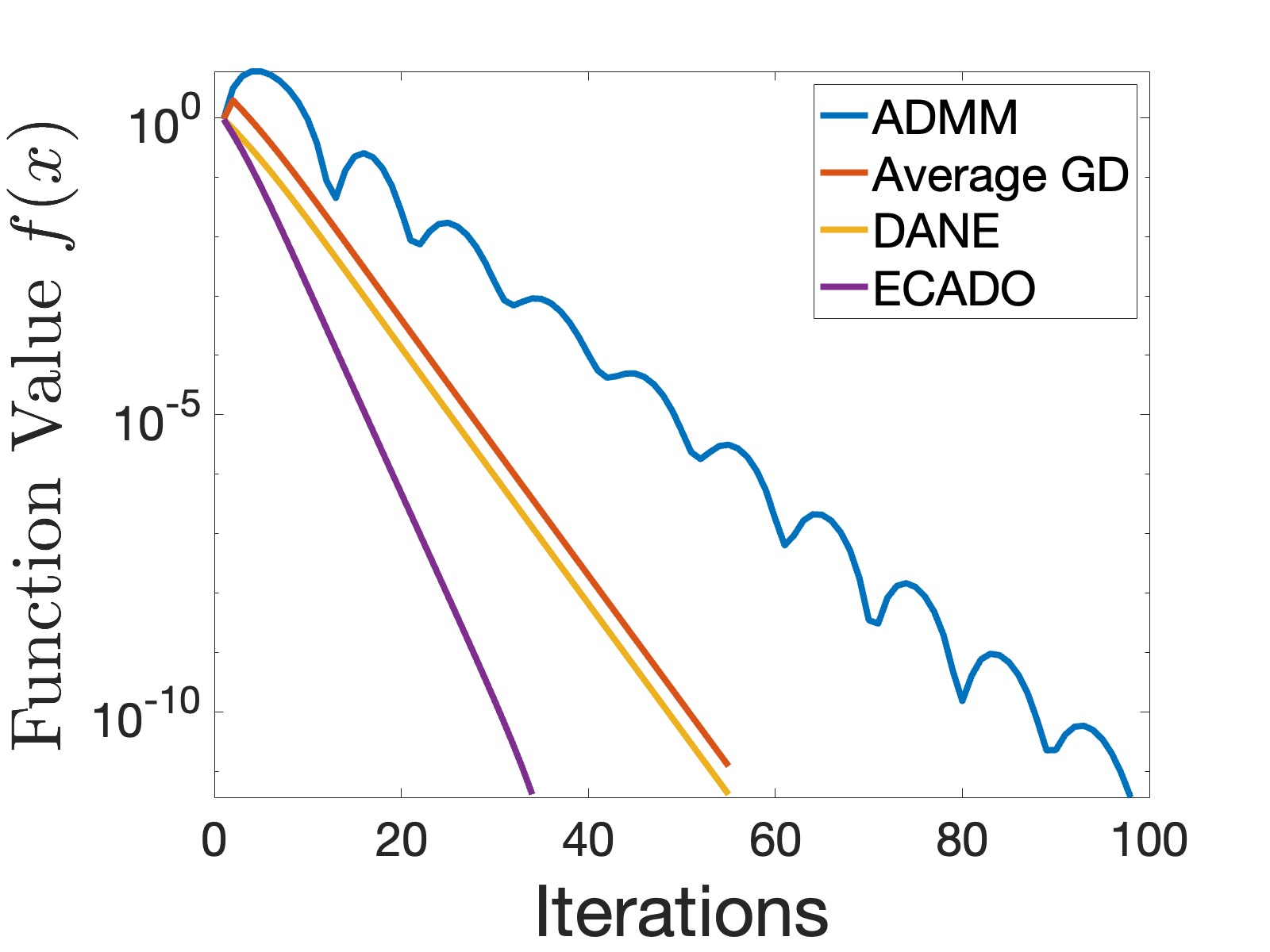}
      \caption{Convergence of Logistic Regression on Synthetic Dataset}
  %\label{fig:logistic_regression_synthetic}
    \end{minipage}%
    \hfill
    \begin{minipage}[t]{.45\linewidth}
      \centering
      \label{fig:logistic_regression_gisette}\includegraphics[width=0.8\columnwidth]{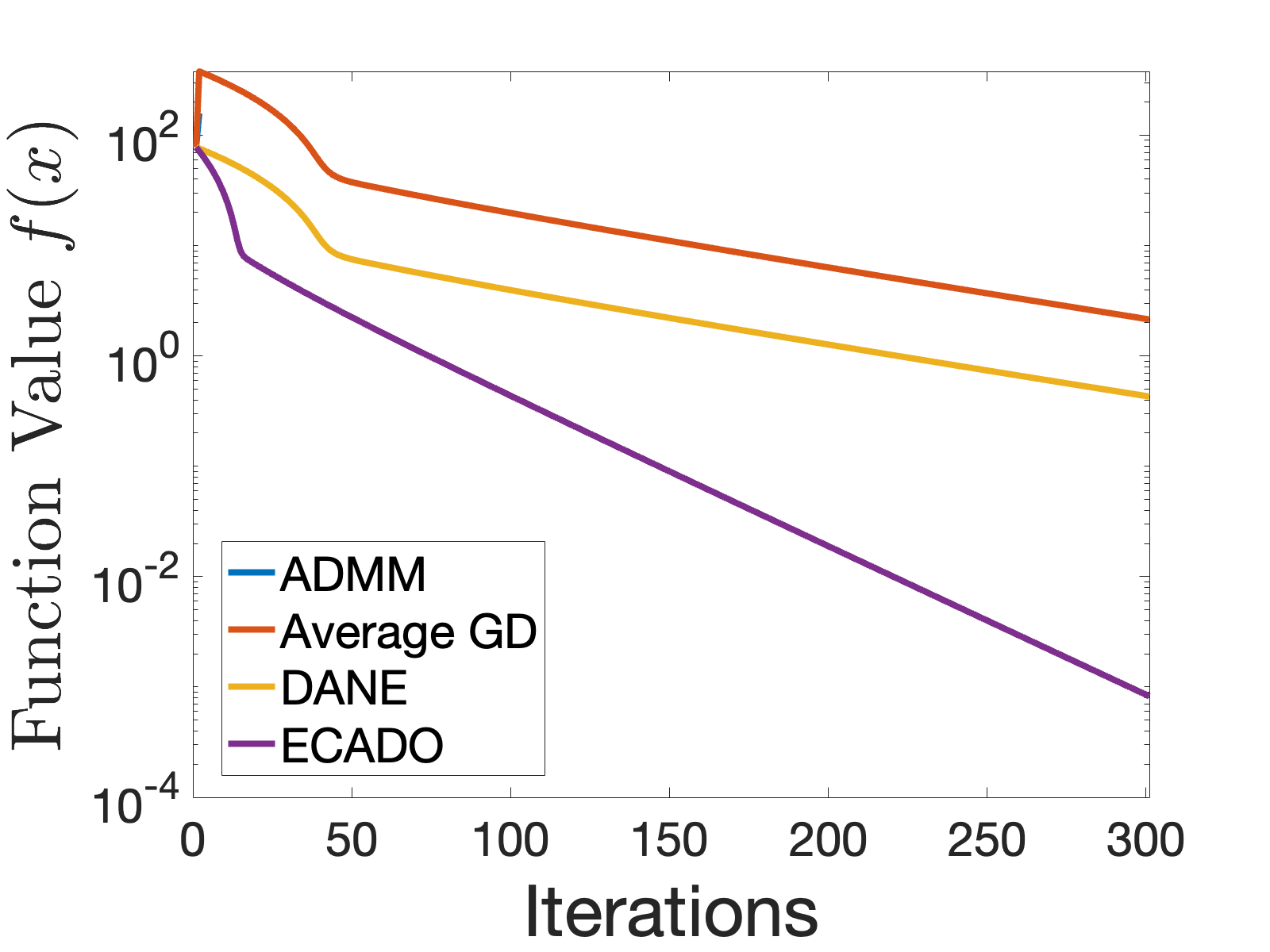}
      \caption{Convergence of Logistic Regression on Gisette Dataset}
      %\label{fig:logistic_regression_gisette}
    \end{minipage}
    \end{figure}

\subsection{Distributed Training Neural Network for Classifying MNIST Data}

ECADO provides guaranteed convergence for convex and nonconvex cases. We study the performance of our algorithm for a non-convex problem by training a 64-neuron neural network to classify MNIST data. This experiment is based on the work in \cite{yuejie_li2020communication}, which trains 60,000 samples across $n=20$ agents. 

The training loss after each central agent update is shown in Figure \ref{fig:nn}. ECADO again provides better convergence than the selected state-of-the-art methods. Since each sub-problem uses an identical $\bar{R}_i^{th}$, for all agents, $i\in[1,n]$, thereby decreasing the benefit of our sensitivity-based weighting scheme. But as with the convex example, ECADO results in faster convergence due to the B.E. steps and the integral controller.

\begin{figure}
    \centering
    \includegraphics[width=0.5\columnwidth]{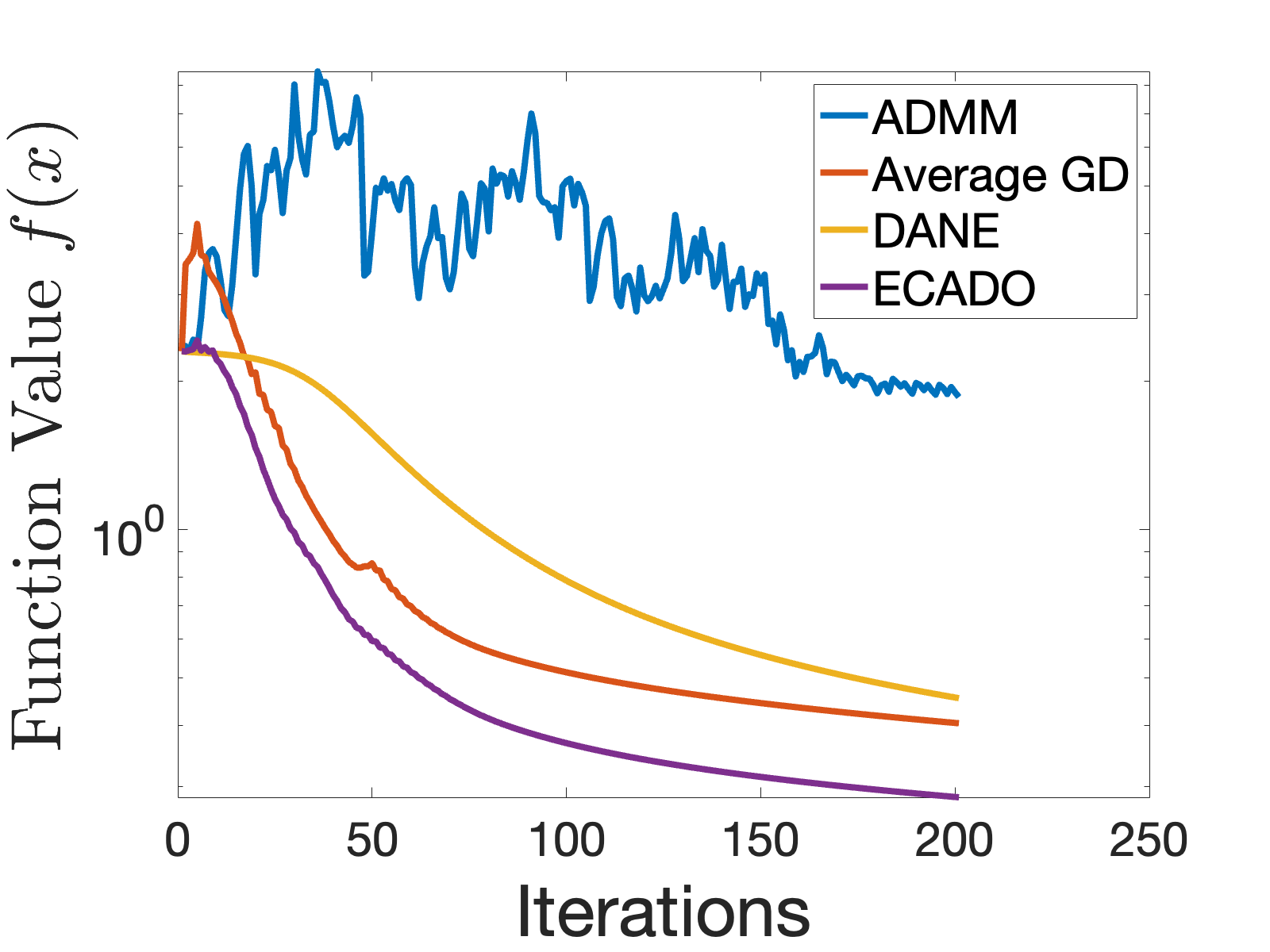}
    \caption{Distributed Training of a 3-Layer Neural Network for Classifying MNIST Data}
    \label{fig:nn}
\end{figure}

\subsection{Distributed Optimization for Security Constrained Optimal Power Flow}

Next, we use our distributed optimization algorithm to solve a large-scale power systems example, known as security contrained optimal power flow (SCOPF). SCOPF is a challenging, multimodal, nonconvex optimization problem designed to find a power grid dispatch, $\x$, (by setting the generator voltages and power generation) that minimize the cost of generation in the dispatched network as well as in $c\in[1,C]$ contingencies. The optimization problem is subject to nonconvex constraints that represent the network constraints of the base dispatch case as well as each ,$c_i$, contingency (possible failures that could occur but for which the grid should remain functional). The optimization problem is described as:
\begin{subequations} \label{eq:scopf}
    \begin{equation}
    \min_{\x} f_b(\x) + \sum_{c=1}^{C}f_c(\x)
    \end{equation}
    s.t.
    \begin{equation}
    g_b(\x) =0 
    \end{equation}
    \begin{equation}
   g_c(\x) =0 \forall c\in [1,C]
    \end{equation}
    \begin{equation}
        h_b(\x) \leq 0
    \end{equation}
    \begin{equation}
        h_c(\x) \leq 0 \forall c \in[1,C]
    \end{equation}
\end{subequations} 

In this optimization problem, $f_b$ represents the cost of generation for the base network, while $f_c$ represents the cost of generation for each contingency. The set of equality constraints, $g_b$ and $g_c$, represent nonlinear network constraints for the base case and contingencies respectively. The nonlinear inequality constraints, $h_b$ and $h_c$, are operational bounds of devices in the grid for the base network and contingencies. Further details on the defining the optimization variables are provided in \cite{go2}.

The challenge with solving \eqref{eq:scopf} is the high-dimensionality and nonconvex constraints caused by the contingency set. The distributed optimization paradigm provides a scalable methodology to solve SCOPF for large networks, whereby the optimization relating to each network $[b,c] \;\; \forall c \in [1,C]$ is optimized on a separate compute node. The centralized agent then provides a consensus on the dispatch decision, $\x$. Additionally, each network has a different topology, thereby incentivizing a weighting scheme that accounts for the variation in the underlying optimization problem. The optimization problem for each network case can be described as:
\begin{subequations} \label{eq:scopf}
    \begin{equation}
    \min_{\x} f_c(\x)
    \end{equation}
    s.t.
    \begin{equation}
   g_c(\x) =0 
    \end{equation}
    \begin{equation}
        h_c(\x) 
    \end{equation}
\end{subequations} 
Each agent solves \eqref{eq:scopf} by constructing a Lagrange $\mathcal{L}_c$ and forming the KKT conditions that are then solved using a Newton-Raphson method. In the context of gradient-flow, this is identical to setting $Z_i=\nabla^2 \mathcal{L}_c$ and applying a Forward-Euler integration with a fixed step size of $\Delta t$. Details of solving each agent is provided in \cite{imb}.

We applied ECADO to a 118-bus SCOPF problem with parameters provided by \cite{go2} and a modified contingency set of 1000 contingencies. The dimensionality of this problem is $d=291,000$. The contingencies are distributed across 40 compute nodes, which each solve 25 contingency networks serially. The updates for each contingency and base network are then communicated to a central compute node that takes a consensus step. The total loss for this optimization problem after each consensus step is shown in Figure \ref{fig:scopf}.

After 20 rounds of of central agent updates, we see that ECADO outperforms all the centralized distributed optimization algorithms. Note that the total network cost in Figure \ref{fig:scopf} is multiplied by $10^6$. The inner optimization of each contingency is performed using the Newton-Raphson-based solver in \cite{imb} that communicates the updates to a central agent. The major benefit of ECADO in this experiment is the precomputed $\bar{R_i}^{th}$ for flat start conditions (where $\x_i=1$). Due to varying tolologies, each contingency network produces a different $\bar{R}_i^{th}$ value, thereby adding variance to the weighting process. %Both ECADO and DANE take advantage of the different structures of each optimization problem, however, ECADO engineers the sensitivity model based on a well-known operating point as well as the integration process. 

\begin{figure}
    \centering
\includegraphics[width=0.5\columnwidth]{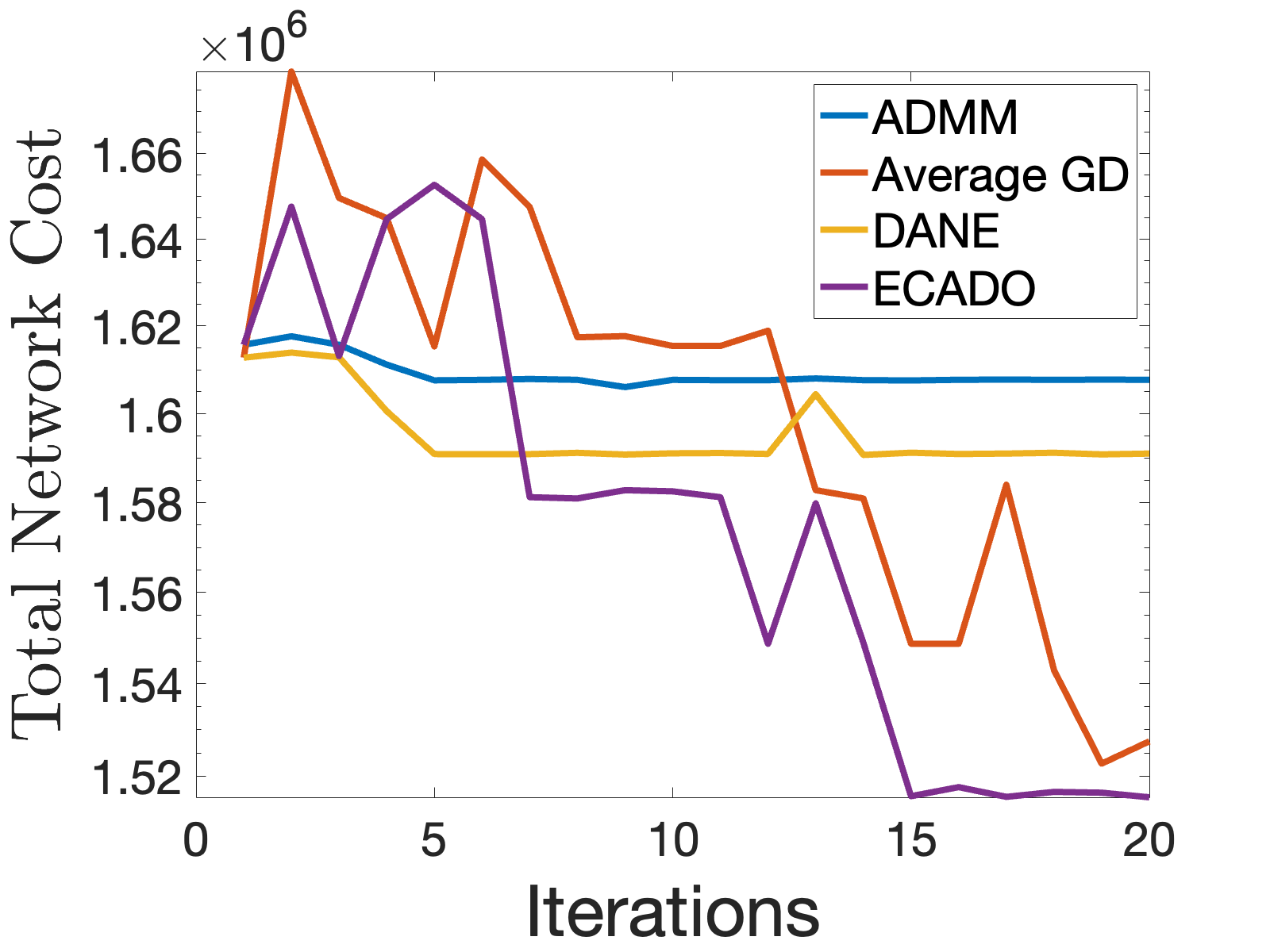}
    \caption{Loss on Distributed Optimization of SCOPF for modified IEEE 118-Bus System}
    \label{fig:scopf}
\end{figure}

\section{Conclusion}

We introduced a new distributed optimization method, ECADO, that was inspired by an equivalent circuit representation of the distributed optimization problem. ECADO uses circuit insights to provide a sensitivity-based weighting scheme and step-size selection based on Backward-Euler numerical integration. Experimentally, ECADO demonstrates faster convergence compared to ADMM, centralized gradient descent, and DANE for a range of convex and non-convex problems that included: solving distributed logistic regression, training a deep neural network, and optimizing high-dimensional security constrained optimal power flow. 

%distributed optimization method that  new modifications to the underlying optimization problem by adding inductor devices to provide a structure better suited for distributed computing. To solve the distributed optimization, we then partition the new equivalent circuit and define a consensus step that uses circuit insights to define a new weighting scheme and step-size selection. The weighting scheme is based on an average linear sensitivity model of each sub-problem, while the step-size selection is based on a Backward-Euler integration. Importantly, each parameter of the method is intuitively designed using circuit principles. The result is a fast converging distributed optimization method that provably converges to a critical point, and can be tuned to achieve faster convergence than state-of-the art methods. Experimentally, ECADO demonstrates superior convergence for both convex and nonconvex applications in comparison to ADMM, centralized gradient descent and DANE.
\newpage
\appendix

\section{EC Model of the Distributed Optimization Problem}

ECADO develops its methods using insights from an equivalent circuit (EC) model of the scaled gradient-flow problem in \eqref{eq:scaled_gd_flow}. The EC represents the continuous-time trajectory of the optimization variable, $\x(t)$, as the transient response of node voltages. For a multi-dimensional optimization problem, $\x\in\R^n$, the EC model is composed of $n$ coupled circuits where each circuit, shown in Figure \ref{fig:ec_model}, models the dynamics of a single variable of $\x(t)$. %Note, in this work, the subscript in the variables refer to a single variable in the vector $\x \in R^n$.%, where as a super-script refers to the variable $x$ on a distributed machine. 

The EC is composed of two circuit elements: a nonlinear capacitor with capacitance of $Z(\x)$ and a voltage controlled current source (VCCS) which produces a current equal to $\nabla f (\x)$. Using Kirchhoff's current law (KCL), the behavior of the EC is expressed as:
\begin{equation}
Z(\x)\dot{\x}(t) + \nabla f(\x) =0,
\end{equation}
which is identical to the scaled-gradient flow \eqref{eq:scaled_gd_flow}. The EC reaches a steady-state when the current through the capacitor is equal to zero ($I_c = Z(\x)\dot{\x}$), at which point, we reach a node voltage, $\x^*$, where
\begin{equation}
    \nabla f(\x^*) =0,
\end{equation}
which is in the set of critical points, $S=\{\x^*|\nabla f(\x^*)=0\}$. 

\begin{figure}
    \centering
    \includegraphics[width=0.5\columnwidth]{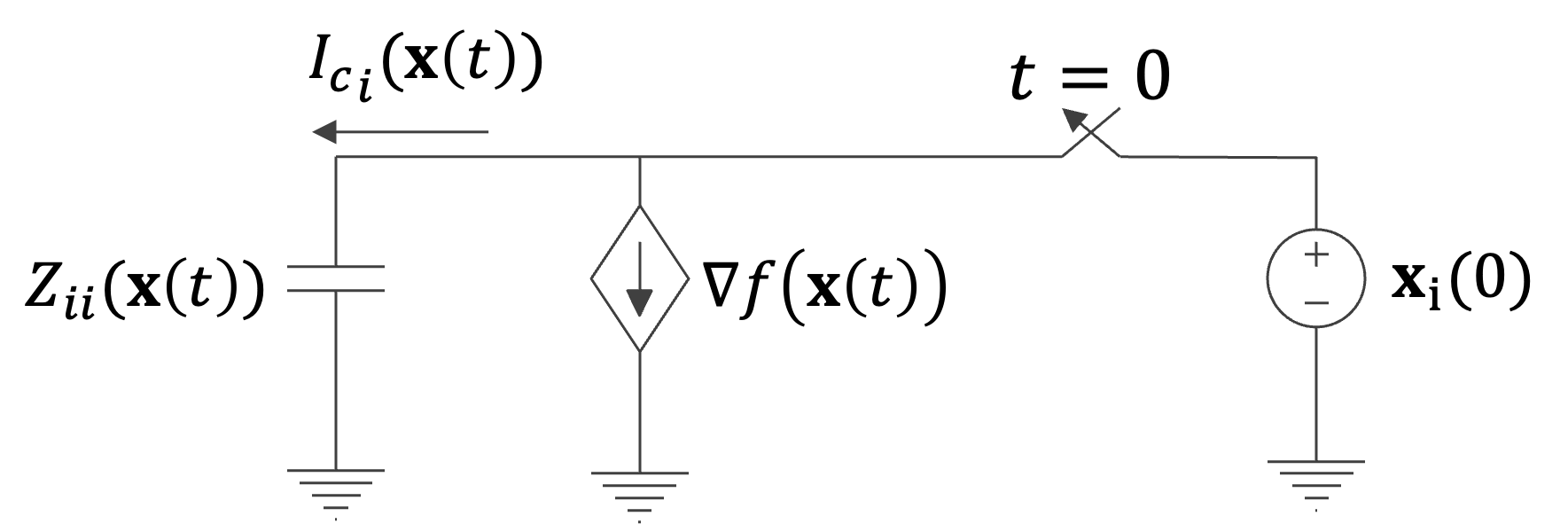}
    \caption{Equivalent Circuit Model of Scaled Gradient Flow \eqref{eq:scaled_gd_flow}}
    \label{fig:ec_model}
\end{figure}

The EC model for the separable scaled-gradient flow \eqref{eq:gd_flow_distributed}, shown in Figure \ref{fig:ec_distributed},  has $m$ VCCS sources in parallel, each producing a current of $\nabla f_{i}(\x)$. The KCL equation characterizing this circuit model is:
\begin{equation}
    Z(\x)\dot{\x} + \sum_{i=1}^{m} \nabla f_{i}(\x) = 0,
\end{equation}
which is identical to the scaled gradient flow of the distributed optimization problem in \eqref{eq:gd_flow_distributed}.

\begin{figure}
    \centering
\includegraphics[width=0.7\columnwidth]{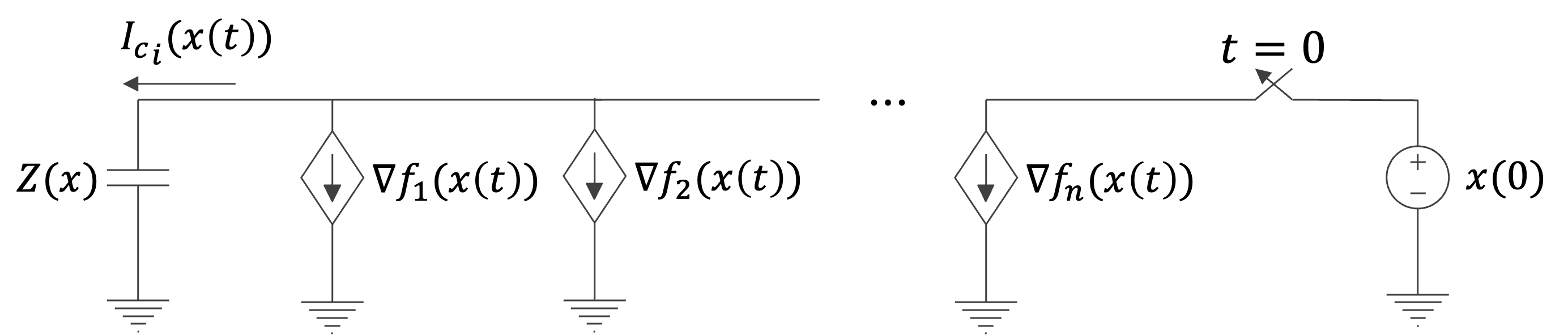}
    \caption{Equivalent Circuit Model of Distributed Optimization Scaled Gradient Flow}
    \label{fig:ec_distributed}
\end{figure}

The EC model in Figure \ref{fig:ec_distributed} is not well-suited for distributed simulation since all VCCS elements share a common node-voltage, $\x$. We design a partitioning scheme using insights from the circuit model to separate the node voltage into $m+1$ vectors.

%To solve the EC model in Figure \ref{fig:ec_distributed}, our goal is to design a partitioning scheme that separates the full EC model into $m+1$ partitions, where $m$ partitions  model the trajectory of each sub-problem, $f_i$ on a separate computing node, and the last partitioned circuit is solved on a centralized agent to provide a consensus amongst all sub-problems. However, the EC model in Figure \ref{fig:ec_distributed} is not easily separable due to a common shared node, $\x$, across all sub-problems and VCCS elements. To develop an EC model that is suited for distribution across multiple computing nodes, we need to apply modifications to partition the circuit by its node-voltages. To achieve this, we perform two modifications to the underlying circuit: 1) separating $Z$ into parallel capacitances, and 2) addition a fictitious small inductor between a centralized node and all sub-circuits. Importantly, these modifications to the EC (and the underlying scaled gradient-flow) are only justifiable through the circuit interpretation and invoke a new structure to the problem that is better suited for distributed computing. Without a physical model, it is unlikely these modifications would be derived, nor would subsequent methods to solve the distributed EC model.

\subsection{Separating the Capacitances} \label{appendix:separating_caps}
The first step in the partitioning scheme separates the nonlinear capacitance, $Z(\x)$, into $m+1$ parallel capacitors:
\begin{equation}
    Z(\x) = Z_c(\x) + \sum_{i=1}^{m} Z_i(\x).
\end{equation}
 The capacitor, $Z_i(\x)$, is associated to each VCCS element, $\nabla f_i(\x)$, and $Z_c$ is a capacitor for the central agent node. This modification draws inspiration from a fundamental circuit principle stating that the total capacitance of capacitors connected in parallel is equal to the sum of their individual capacitances. The resulting EC model is shown in Figure \ref{fig:ec_distributed_caps}.
 
 %, has $m+1$ parallel capacitors and is redrawn to group the capacitor for each sub-problem, $Z_i$ with the corresponding VCCS element $\nabla f_i$. Importantly, in redrawing the capacitors as multiple parallel capacitances, the trajectory of the state variables in EC in Figure \ref{fig:ec_distributed_caps} remains unmodified from the circuit perspective. 

\begin{figure}
    \centering
\includegraphics[width=0.9\columnwidth]{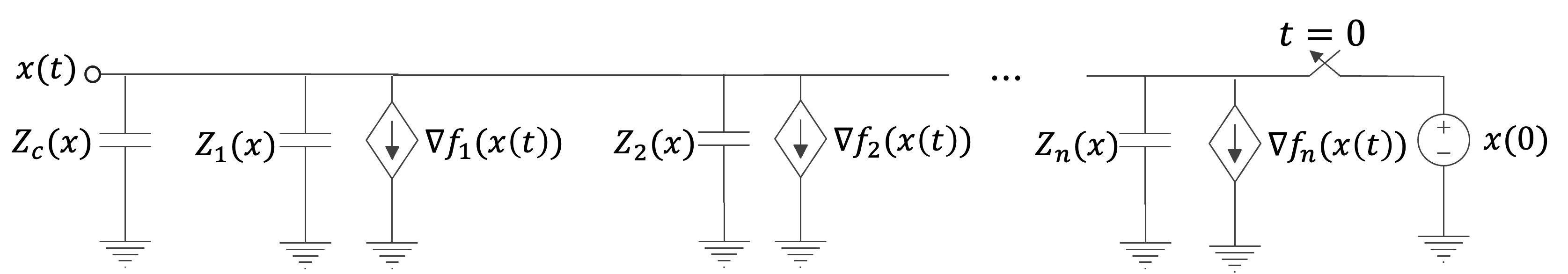}
    \caption{Modifying Capacitors for Equivalent Circuit Model of Distributed Optimization}
    \label{fig:ec_distributed_caps}
\end{figure}

Furthermore, separating $Z(\x)$ into parallel capacitors does not affect the steady-state of the EC. At the DC steady-state ($\x^*$), the capacitors in the EC look like an open circuit (i.e., with zero current flowing) described as:
\begin{align}
    Z_i(\x*)\dot{\x^*}=0, \;\; \forall i\in[1,m] \\ 
    Z_c(\x^*)\dot{\x^*}=0.
\end{align}
By KCL, this implies that the sum of the VCCS currents equals zero, $\sum_{i=1}^{m}\nabla f_{i}(\x^*) =0 $, for which the solution is in the set $S$.

%This process to create multiple capacitances is the first step in defining separable sub-circuits, each with its own capacitor model, $Z^i(\x)$. Importantly, this modification does not modify the steady-state of the original EC model (i.e., the critical points of the objective function).

\subsection{Adding a small Inductor}
\label{appendix:adding_inductor}
The second step in the paritioning scheme introduces an inductor between a central agent node, $\x_c$, and each VCCS element to define a new node, $\x_i\; \forall i\in[1,m]$, as shown in Figure \ref{fig:ec_distributed_inductors}. This step maps the flow variable, $I_i^L$, introduced in Section \ref{sec:partitioning_ec}, to a physical inductor model with an inductance of $L$, which produces a current, $I_i^L$, according to:
\begin{equation}
    L \dot{I}_i^L = \x_c - \x_i.
\end{equation}

\begin{figure}
    \centering
    \includegraphics[width=0.5\columnwidth]{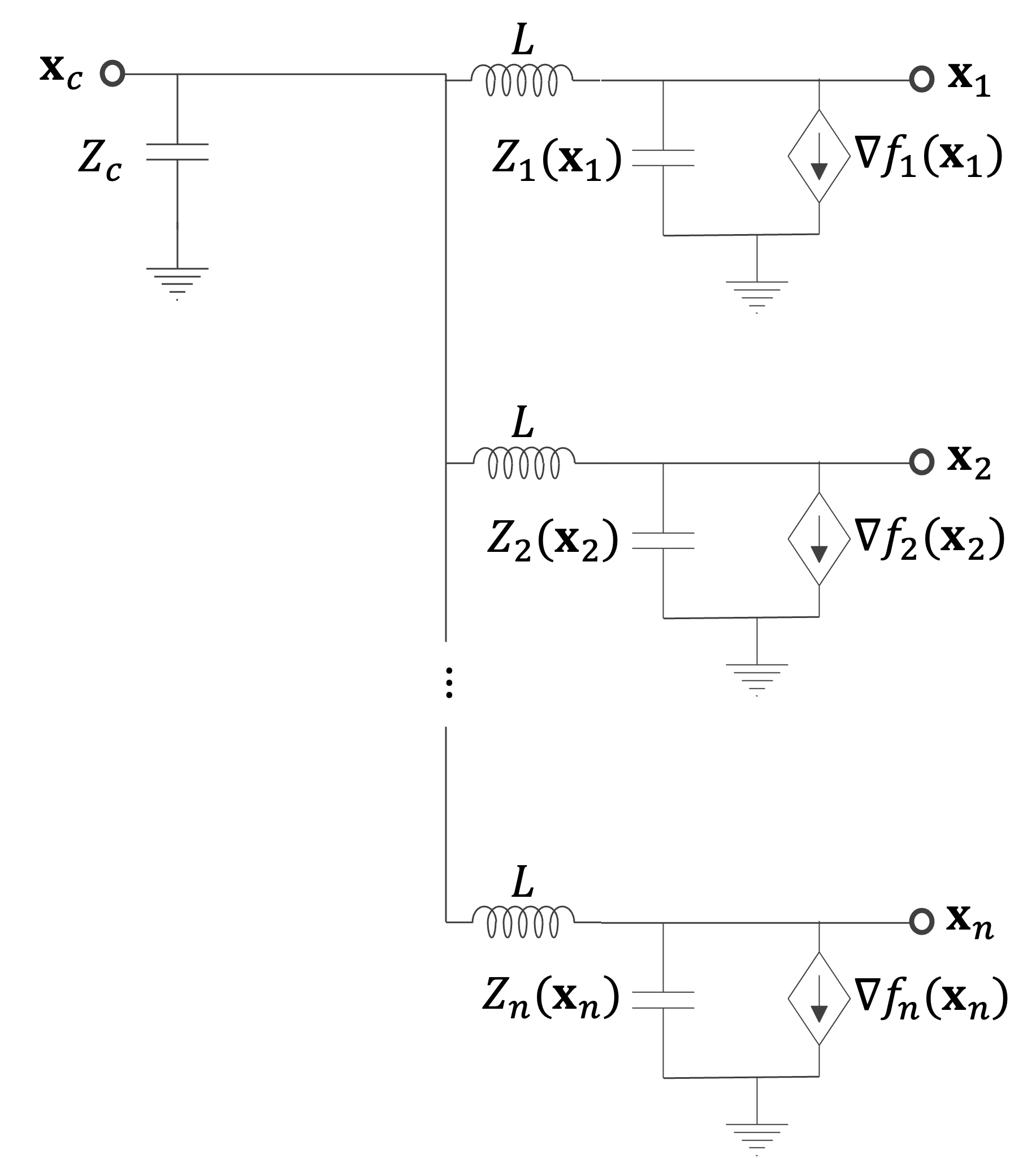}
    \caption{Modifying Equivalent Circuit Model for Distributed Optimization by Adding Series Inductors}
    \label{fig:ec_distributed_inductors}
\end{figure}

The inductors separate the EC nodes so that each VCCS element, $\nabla f_i (\x_i)$, and nonlinear capacitor, $Z_i(\x_i)$, is controlled by a local node-voltage, $\x_i$. Additionally, each sub-circuit is coupled to a central agent node, $\x_c$, through the inductor current. This EC model is better suited for distributed optimization, as each VCCS and capacitance matrix can be defined by a local node-voltage. 

% The inductor separates the central agent from each sub-problem. In this modified structure, the variable $\x$ is separated into multiple variables, $\x_c\in\R^n$ and $\x_i\in\R^n$. This modifies the structure of the problem to be better suited for distributed simulation.  Adding these inductors modifies the underlying scaled gradient-flow equations to \eqref{eq:gd_flow_inductor_cenrtal}, \eqref{eq:gd_flow_inductor_inductor},\eqref{eq:gd_flow_inductor_subproblem}.
 \iffalse
 \begin{align}
     &Z_c(\x_c)\dot{\x_c}(t) + \sum_{i=1}^{m}I_i^L =0 \\
     &L\dot{I}_i^{L} = \x_c - \x^i \\
     & Z^i(\x^i)\dot{\x}^i(t) + \nabla f^i(\x^i) - I_L^i = 0
 \end{align}
\fi
 %The VCCS elements and nonlinear capacitances for each sub-problem is now a function of the individual node voltage, $\x_{i}$, thereby disassociating each individual sub-problem. 

%The inductors act as additional energy-storage elements between the centralized node and each individual agent node, $\x_i$. From an optimization perspective, it is shown in \cite{boyd} that the inductor provides a second order effect similar to acceleration methods in optimization. %To ensure that the second-order effect caused by the inductors do not dominate the response by creating oscillations in the continuous-time trajectory, we design $l<<min(Z^i)$.

Note, the addition of the inductors does not affect the steady-state solution of the EC. At the DC steady-state of the EC, the inductors are effectively shorted, $L \dot{I}_i^L=0$, while the capacitors are open; this defines the steady-state as:
\begin{align}
    \x_c -\x_i =0 \;\; \forall i \in [1,m] \\
    \nabla f_i(\x_i)=0 \;\; \forall i\in[1,m]\\
    \implies \sum_{i=1}^{m} \nabla f(\x_c)=0,
\end{align}
 which is equivalent to the steady-state condition of the scaled gradient-flow.

\iffalse
\begin{theorem}
The modified EC model in Figure XX of the distributed optimization problem \eqref{eq:gd_flow_distributed}, achieves the same steady-state as the set of critical points, $S$.
\end{theorem}

\begin{proof}
In steady-state, the inductor current is not changing, and the voltage difference across the inductor is zero. For the EC model, this equates to
\begin{equation}
    L\dot{I}_L^i = \x_c - \x^i \equiv 0.
\end{equation}

The entire EC model at steady-state is characterized by zero capacitor currents and zero voltage difference across the inductors. The set of equations describing the EC at steady-state is:
\begin{align}
    \sum_{i=1}^{n} I_L^i &=0 \\
    \x_c - \x^i & =0 \forall i\in n \\
    \nabla f^i(\x^i) -I_L^i &=0 \forall i \in n
\end{align}
In rearranging the set of equations to eliminate $I_l^i$ and set $\x^i=\x_c$, the steady-state of the EC model is
\begin{equation}
    \sum_{i=1}^{n} \nabla f^i(\x_c) =0,
\end{equation}
which is equivalent to the steady-state condition (XX).
\end{proof}
\fi
%The main contribution of adding the inductors is to separate the nodes into $n+1$ nodes. The system of equations describing the

\section{Solving the Equivalent Circuit Model Using Gauss-Seidel} \label{sec:partitioning_ec}

To solve the EC model in Figure \ref{fig:ec_distributed_inductors}, we distribute the computation using a Gauss-Seidel method (Algorithm \ref{basic-gauss-seidel}) that divides the circuit into $m+1$ sub-circuits wherein each sub-problem EC, shown in Figure \ref{fig:ec_subcircuit}, is decoupled from the central agent EC depicted in Figure \ref{fig:partition_circuit}. In each G-S iteration, the EC representing the sub-problem solves for the node-voltage, $\x_i^{k+1}$, by utilizing a constant current source to model the inductor current, $I_i^{L^k}$. 
The central agent then uses a constant voltage source with a voltage of  $\x_i^{k+1}$ to represent each sub-problem EC and solve  for the updated inductor currents, $I_i^{L^{k+1}}$, and node-voltage, $\x_c^{k+1}$. At the end of the G-S iteration, the updated values for the inductor currents are communicated back to the sub-circuits in Figure \ref{fig:ec_subcircuit}.  Analyzing the G-S iterations as an exchange of voltages and currents in the EC model allows us to develop circuit-inspired methods to improve the convergence of the G-S algorithm.

Namely, ECADO improves solving the central agent EC by modeling each sub-circuit using a Thevenin model consisting of a Thevenin resistance, $R_i^{th}$, and Thevenin voltage source, $\x_i^{th}$. The Thevenin resistance models the linearized sensitivity looking into each sub-circuit as a one-port resistance to provide a more accurate representation of each sub-EC. The Thevenin resistance is calculated in Appendix \ref{sec:deriving_rth} and implemented into the G-S in Algorithm \ref{full-gauss-seidel}.

%Distributing the modified equivalent circuit model in Figure \ref{fig:ec_distributed_inductors} using a Gauss-Seidel method can be viewed as partitioning the circuit using a node-tearing method \cite{lawrence1995electronic}. In this approach, the full circuit model is partitioned into $m+1$ sub-circuits, where $m$ sub-circuits represent the differential equations for the sub-problem in \eqref{eq:gd_flow_inductor_subproblem} and the last sub-circuit represents the central agent differential equation \eqref{eq:gd_flow_inductor_cenrtal},\eqref{eq:gd_flow_inductor_inductor}. At each iteration of the Gauss-Seidel iteration, the inductor current from the previous iteration of the central agent (represented by $I_i^{L^k}$) is treated as a constant current source input to each sub-problem circuit as shown in Figure \ref{fig:ec_subcircuit}. The sub-circuit is simulated for $t\in[t_1,t_2]$ and the resulting voltage waveform, $\x_i(t)$ is then passed to the central agent. The central agent receives the voltage waveform for all sub-problems, and in the Gauss-Seidel fashion, models the sub-circuit as a constant voltage source. This process is an equivalent circuit representation of the Gauss-Seidel process in \eqref{eq:gd_flow_inductor_cenrtal}.

\begin{figure}
    \centering
    \includegraphics[width=0.5\columnwidth]{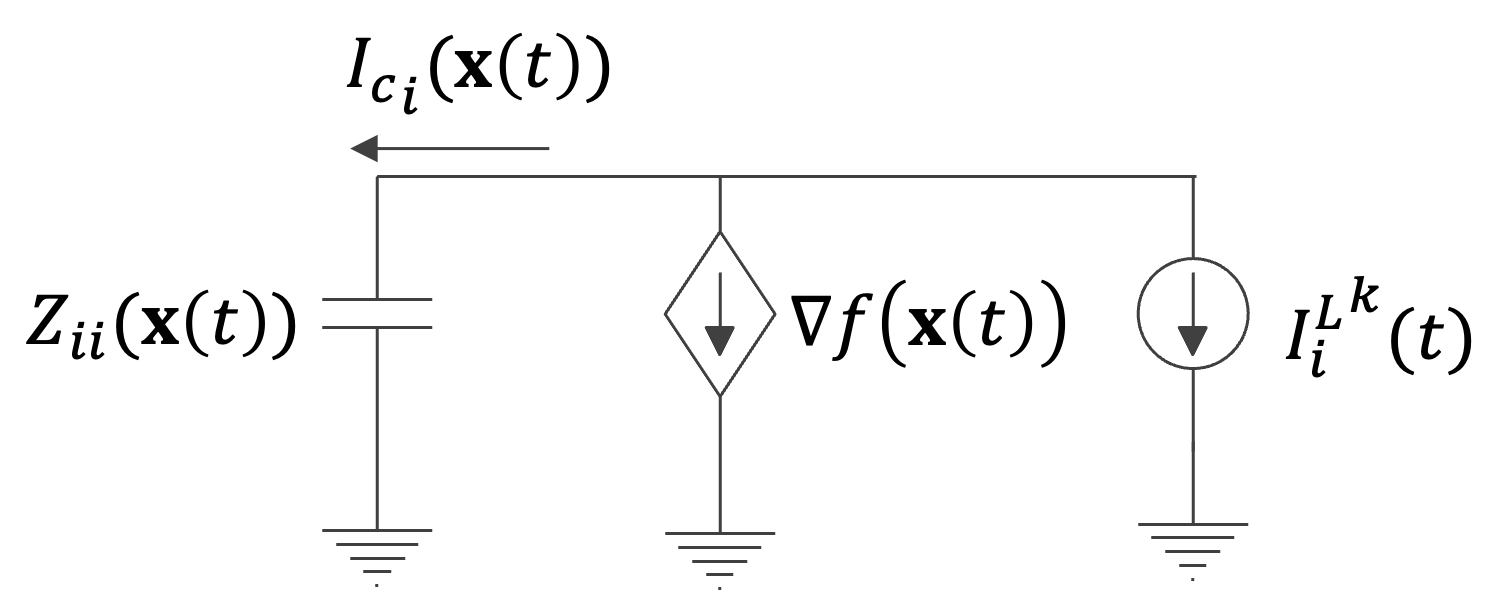}
    \caption{The partitioned equivalent circuit model for each sub-problem}
    \label{fig:ec_subcircuit}
\end{figure}

\begin{figure}
    \centering
    \begin{minipage}[t]{.45\linewidth}
      \centering
      \includegraphics[width=0.7\columnwidth]{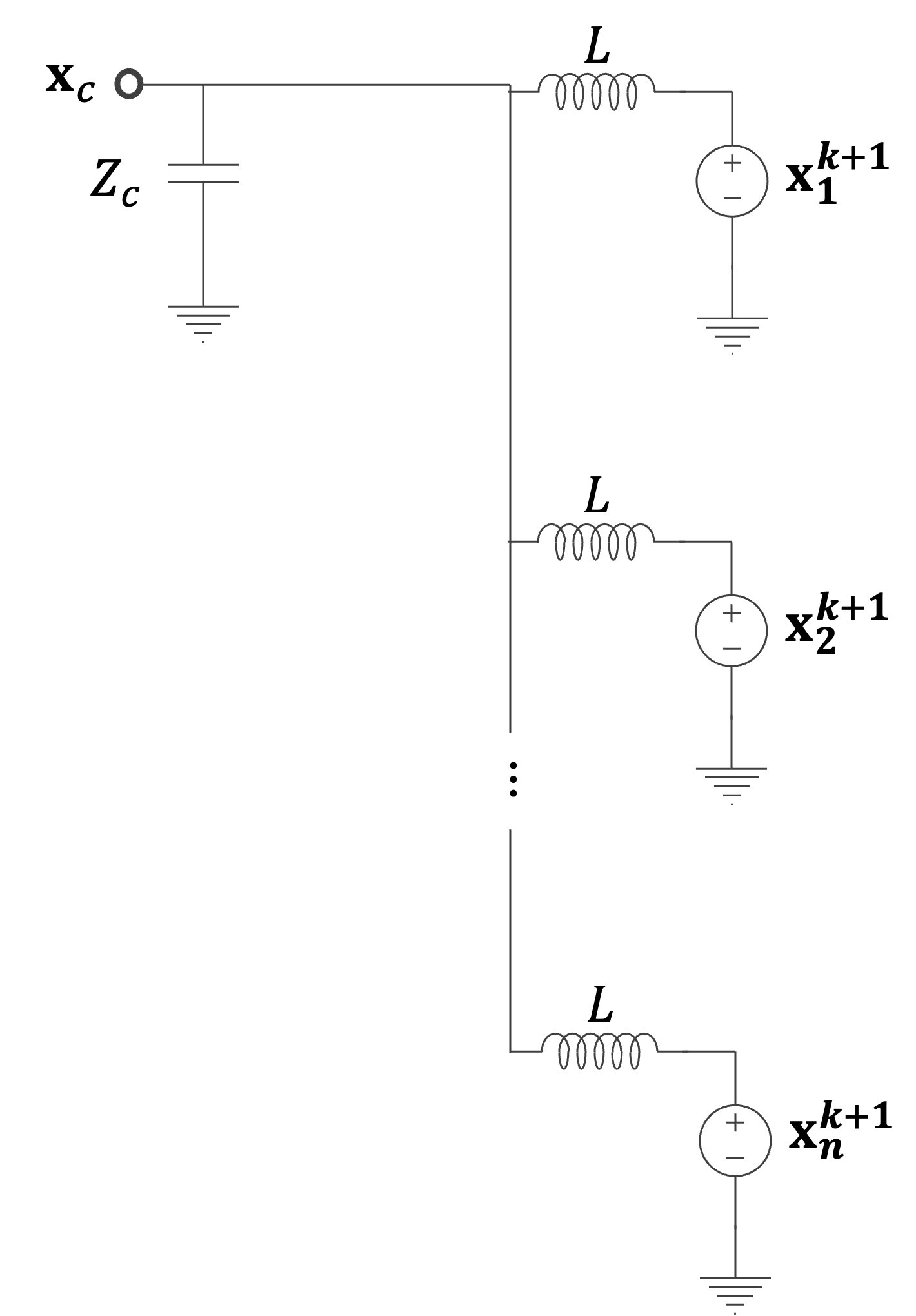}
      \caption{The partitioned equivalent circuit model for the central agent ODE \eqref{eq:gd_flow_inductor_cenrtal},\eqref{eq:gd_flow_inductor_inductor}}
  \label{fig:partition_circuit_subproblem}
    \end{minipage}%
    \hfill
    \begin{minipage}[t]{.45\linewidth}
      \centering
      \includegraphics[width=0.8\columnwidth]{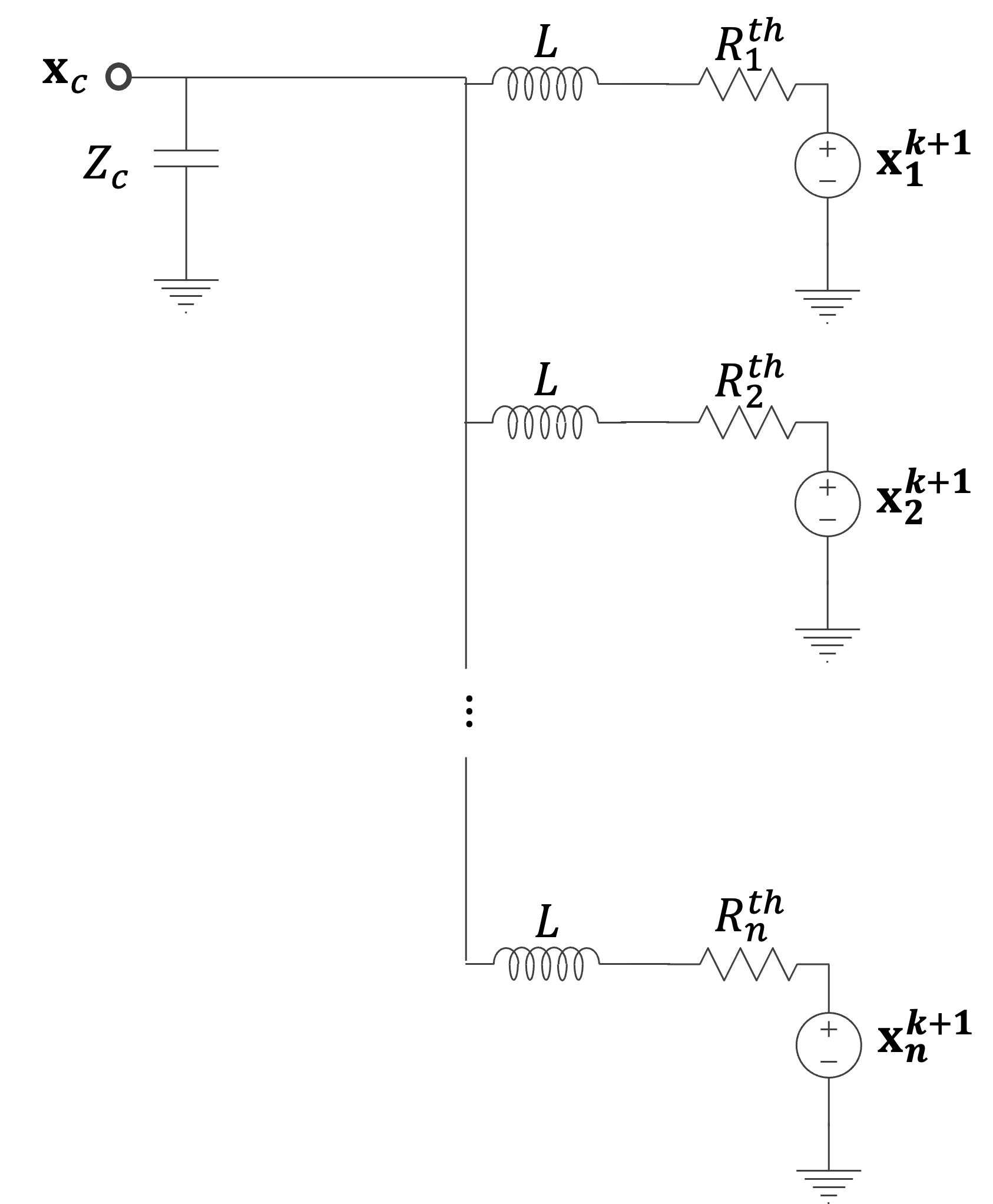}
      \caption{The partitioned equivalent circuit model for the central agent ODE \eqref{eq:central_chord_ode} with a Thevenin Equivalent model}
      \label{fig:partition_circuit}
    \end{minipage}
\end{figure}

%The constant voltage source in the central agent models the behavior of each sub-problem over the time-period. A more accurate model is a \emph{Thevenin-equivalent} model which consists of a thevenin voltage source and a thevenin resistance to represent the linear sensitivity of each sub-problem. The Thevenin model allows us to reduce an entire sub-problem to a one-port model consisting of a sensitivity term, $R_i^{th}$ and a Thevenin voltage source. At the operating point, $\x_i^{k+1}$, the Thevenin model exactly represents the sub-circuit through a substitution principle \cite{lawrence1995electronic}. The added Thevenin resistance models the linearized sensitivity as a one-port resistance looking into each sub-circuit.

%The Thevenin resistance can be calculated using the method provided in Appendix \ref{sec:deriving_rth}.

\section{Deriving the Linear Sensitivity}
\label{sec:deriving_rth}
The linear sensitivity model, $R_i^{th}$, measures the sensitivity of the local node, $\x_i$, to a perturbation in the inductor current, $I_i^L$. The change in the node-voltage, $\Delta \x_i$, due to a perturbation in the inductor current, $\Delta I_i^L$, can be calculated as:
\begin{align}
    Z_i(\x_i + \Delta \x_i) (\dot{\x}_i + \dot{\Delta \x}_i) + \nabla f_i(\x_i + \Delta \x_i) - (I_i^L + \Delta I_i^L) = 0 \\
\end{align}
We approximate the sensitivity using the first two terms of the Taylor series:
\begin{equation}
    (Z_i(\x_i) + \frac{\partial Z_i(\x_i)}{\partial \x_i}\Delta \x_i)(\dot{\x_i} + \Delta \dot{\x_i}) + \nabla f_i(\x_i) +  \frac{\partial \nabla f_i(\x_i)}{\partial \x_i}\Delta \x_i - I_i^L - \Delta I_i^L = 0. \label{eq:Rth_derivation_1}
\end{equation}
which, knowing that $Z_i(\x_i) + \nabla f_i(\x_i) - I_i^L=0$, is reduced to:
\begin{equation}
         (Z_i(\x_i) + \frac{\partial Z_i(\x_i)}{\partial \x_i}\Delta \x_i) \Delta \dot{\x_i} + \frac{\partial Z_i(\x_i)}{\partial \x_i}\dot{\x_i}+  \frac{\partial \nabla f_i(\x_i)}{\partial \x_i}\Delta \x_i - \Delta I_i^L = 0. \label{eq:Rth_derivation_2}
\end{equation}

To simplify the derivation, we assume a fixed capacitance model, i.e., $\frac{\partial Z_i(\x_i)}{\partial \x_i}=0$, to define a linear sensitivity equation:
\begin{equation}
             Z_i(\x_i)\Delta\dot{\x}_i +  \frac{\partial \nabla f_i(\x_i)}{\partial \x_i}\Delta \x_i - \Delta I_i^L = 0. \label{eq:Rth_derivation_3}.
\end{equation}

The ODE above is then solved using a Backward-Euler integration:
\begin{equation}
    \Delta \x_i (t+\Delta t) = \Delta \x_i(t) +\Delta t Z_i^{-1} \left(  I_i^L(t+\Delta t) - \frac{\partial \nabla f_i(\x_i)}{\partial \x_i}\Delta \x_i(t+\Delta t) \right)
\end{equation}
from which we can derive $R_i^{th}=\Delta x_i / \Delta I_i^L$ as :
\begin{equation}
    R_i^{th} = \frac{\Delta \x_i}{\Delta I_i^L} = (\frac{Z_i}{\Delta t}+\frac{\partial \nabla f_i(\x_i)}{\partial \x_i})^{-1}
\end{equation}

\section{Proof of Convergence}
\label{sec:convergence_proof_appendix}

To prove convergence of G-S, we abstract the system of equations \eqref{eq:ec_ode_matrix} as:
\begin{equation}
    C(X)\dot{X}(t) = \hat{f}(X),
\end{equation}
where $X = [\x_1,\x_2,\ldots, \x_n, I_1^L, I_2^L, \ldots, I_n^L, \x_c]$. The matrix $C(X)$ and vector $\hat{f}(X)$ are defined as:
\begin{equation}
    C(X)=\begin{bmatrix}
        Z_1(\x_1) &0 & \ldots & 0 & 0 & \ldots & 0\\
        0 & Z_2(\x_2) &\ldots & 0 & 0 &\ldots &0\\
        0 & 0 & \ddots & 0 & 0 &\ldots & 0 \\
        0 & 0 & \ldots & L & 0 & \ldots & 0\\
        0 &0 & \ldots & 0 & L & \ldots & 0\\
        0 &0 & \ldots & 0 & 0 & \ddots & 0\\
        0 &0 & \ldots & 0 & 0 & \ldots & Z_c
        \end{bmatrix}, \;\;\; \hat{f}(X) = 
        \begin{bmatrix}
            0 & 0 & \ldots & 1 & 0 & \ldots & 0\\
            0 & 0 & \ldots & 0 & 1 &\ldots & 0\\
            0 & 0 & \ddots &0 & 0 & \ddots & 0 \\
            -1 & 0 & \ldots & -\bar{R}_1^{th} & 0 & \ldots &1\\
            0 & -1 & \ldots & 0 & -\bar{R}_2^{th} & \ldots & 1 \\
            0 & 0 & \ddots & 0 & 0 &\ddots & 1\\
            0 & 0 & \ldots &-1 & -1 & \ldots &0
        \end{bmatrix}
        \begin{bmatrix}
            \x_1 \\
            \x_2 \\
            \vdots \\
            I_1^L \\
            I_2^L \\
            \vdots \\
            \x_c
        \end{bmatrix} + \begin{bmatrix}
            -\nabla f_1(\x_1) \\
            -\nabla f_1(\x_2) \\ 
            \vdots \\
            I_1^{L^k} \bar{R}_1^{th}\\
            I_2^{L^k} \bar{R}_2^{th} \\
            \vdots \\
            0
        \end{bmatrix}.
\end{equation}

 $C(X)$ is defined as a positive diagonal-matrix, which is a sufficient condition to prove convergence.

 \begin{theorem}
    \label{thm:convergence}
    Given a positive diagonal-matrix, $C(X)$, the Gauss-Seidel sequence $\{\dot{X}^k\}$  converges to the fixed point of \eqref{eq:ec_ode_matrix}, where each iteration satisfies:
    \begin{equation}
        \| \dot{X}^{k+1} - \dot{X}\| \leq \alpha \| \dot{X}^k - \dot{X}\|,
    \end{equation}
    for $\alpha < 1$.
 \end{theorem}

\begin{proof}
    We prove that the Gauss-Seidel converges to a fixed point for the general case of a strictly-diagonally dominant matrix, $C(X)$, for which ECADO is a special case. This proof is reconstructed from \cite{white2012relaxation}.

    Let us decompose the matrix, $C(X)$, into diagonal $D(X)$, upper triangular, $U(X)$, and lower triangular, $L(X)$, components as:
    \begin{equation}
        C(X) = D(X)+L(X)-U(X).
    \end{equation}

        Assuming strict diagonal dominance, we conclude that:
        \begin{equation}
            \| (D(X) + L(X))^{-1}U(X) \|_{\infty} < 1, \label{eq:diagonal_dominance_matrix_ineq}
        \end{equation}
        which is proven using standard matrix theory \cite{white2012relaxation}.

        The proof for Theorem \ref{thm:convergence} requires the following two lemmas, which are provide in \cite{white2012relaxation}.
        \begin{lemma}
            Let $Y$ be a Banach space and $F: Y\rightarrow Y$. If $\|F(a) - F(b)\| \leq \gamma \|a-b\| \forall [a,b] \in Y, \gamma \in [0,1)$, then there exists a fixed point, $\bar{y}$ such that $F(\bar{y}) = \bar{y}$.  \label{lemma:contraction}
        \end{lemma}

        \begin{lemma}
            \label{lm:norm_ineq}
            For two vectors, $W,Z \in \Re^{p}$, where $p$ is the dimension of $X$, if there exists a relation: 
            \begin{equation}
                \|\dot{W}(t)\| \leq \gamma \|\dot{Z}\| + l_1 \|W(t)\| + l_2 \|Z(t)\|,
            \end{equation}
            for positive scalars $l_1,l_2 < \inf$ and $\gamma<1$, then we can conclude there is a norm $\| \|_b$ that:
            \begin{equation}
                \|\dot{W}(t)\|_b \leq \alpha \| \dot{Z}\|_b + l_1 \|W(0)\| + l_2 \|Z(0)\|.
            \end{equation}
        \end{lemma}

        The G-S process generates a sequence of vectors, $\{X^{k}\}$, where at iteration $k+1$, the matrix, $C(X)$, is decomposed into:
        \begin{equation}
            C(X^{k+1}) = D^{k+1} + L^{k+1} - U^{k+1}.
        \end{equation}
        In matrix form, the G-S iteration is defined as:
        \begin{equation}
            (L^{k+1} + D^{k+1})\dot{X}^{k+1} - U^{k+1}\dot{X}^k = \hat{f}(X^{k+1}, X^{k}),
        \end{equation}
        which solves for the state vector, $\dot{X}$, as:
        \begin{equation}
            \dot{X}^{k+1} = (L^{k+1} + D^{k+1})^{-1}(U^{k+1}\dot{X}^k + \hat{f}(X^{k+1}, X^{k})).
        \end{equation}       

        To prove convergence to a fixed point, we analyze the difference between $k+1$ iteration and $j+1$ iteration:
        
        \begin{multline}
            \dot{X}^{k+1} - \dot{X}^{j+1} = (L^{k+1} + D^{k+1})^{-1}U^{k+1}\dot{X}^k - (L^{j+1} + D^{j+1})^{-1}U^{j+1}\dot{X}^j  \\ + (L^{k+1} + D^{k+1})^{-1}\hat{f}(X^{k+1}, X^k) - (L^{j+1} + D^{j+1})^{-1}\hat{f}(X^{j+1}, X^j).
        \end{multline}

        Assuming that $\hat{f}$ is Lipschitz continuous with a Lipschitz factor of $l_1$ and $C(X)$ is strictly diagonally dominant, the difference between the $k+1$ and $j+1$ iteration is bounded by:
        \begin{multline}
            \| \dot{X}^{k+1} - \dot{X}^{j+1} \| \leq \|(L^{k+1}+D^{k+1})^{-1}U^{k+1}\dot{X}^k - (L^{j+1}+D^{j+1})^{-1}U^{j+1}\dot{X}^j\| \\ +  l_1K\|X^{k+1} - X^{j+1}\| + l_1 K \|X^k = X^j\|,
        \end{multline}
        where  $K\leq \infty$ such that $\|(L^{k+1} + D^{k+1})^{-1}\| \leq K$.
        
        By \eqref{eq:diagonal_dominance_matrix_ineq}, we can also select a $\gamma$, independent of iteration $k$, such that:
        \begin{equation}
            \|(L^k + D^k)^{-1}U^k\| <  \gamma < 1.
        \end{equation}

        This implies that 
        \begin{equation}
            \| \dot{X}^{k+1} - \dot{X}^{j+1} \| \leq \gamma \|\dot{X}^k - \dot{X}^j\| +  l_1K\|X^{k+1} - X^{j+1}\| + l_1 K \|X^k-X^j\|.
        \end{equation} 
        
        From Lemma \ref{lm:norm_ineq}, we can conclude that 
        \begin{equation}
            \| \dot{X}^{k+1} - \dot{X}^{j+1} \| \leq \gamma \|\dot{X}^k - \dot{X}^j\| +  l_1K\|X^{k+1}(0) - X^{j+1}(0)\| + l_1 K \|X^k(0)-X^j(0)\|.
        \end{equation}
        which is reduced to the following inequality:

        \begin{equation}
            \| \dot{X}^{k+1} - \dot{X}^{j+1} \| \leq \gamma \|\dot{X}^k - \dot{X}^j\|,
        \end{equation} 
        since both sequences $k$ and $j$ are initialized at fixed initial conditions $X(0)=X_0$.
        
        Through the contraction mapping in Lemma \ref{lemma:contraction}, this implies that there exists a fixed point solution for which:
        \begin{equation}
            \dot{X}^{k+1} = \dot{X}^k,
        \end{equation}
        thus proving convergence of G-S.
\end{proof}

\section{Convergence Rates}

In this section, we use a Gauss-Seidel analysis to study the convergence rate of centralized gradient descent (CGD), DANE and ADMM for optimizing the following strongly convex, quadratic, separable objective function:
\begin{align}
    \sum_{i=1}^{m} f_i(\x) \\
    f_i(\x) = \frac{1}{2} A_i\|\x\|^2 + B_i \x + C,
\end{align}
where $A,B\in \Re^{m,m}$ and $C\in \Re$. The gradient of each sub-problem is:
\begin{align}
    \nabla f_i(\x) = A_i\x + B_i.
\end{align}

\subsection{Convergence Rate of Centralized Gradient Descent}
The update using CGD with a fixed-step size of $\alpha$ is:
\begin{align}
    \x_i^{k+1} = \x_c^{k} - \alpha (A_i\x_c + B_i) \;\;\forall i\in[1,n]\\
    \x_c^{k+1} = \frac{1}{n}\sum_{i=1}^{n} \x_i.
\end{align}
Each iteration of CGD in a matrix form is:
\begin{align}
    &\begin{bmatrix}
        1 & 0 & \ldots & 0 \\
        0 & 1 & \ldots & 0 \\
        0 & 0 & \ddots & 0 \\
        -1/n & -1/n & \ldots & 1
    \end{bmatrix} \begin{bmatrix}
        \x_1^{k+1} \\ \x_2^{k+1} \\ \vdots \\ \x_c^{k+1}
    \end{bmatrix} = \begin{bmatrix}
        0 & 0 & \ldots & 1-\alpha A_1 \\
        0 & 0 & \ldots & 1-\alpha A_2 \\
        0 & 0 & \ddots & 1-\alpha A_i \\
        0 & 0 & \ldots & 0
    \end{bmatrix} \begin{bmatrix}
        \x_1^{k+1} \\ \x_2^{k+1} \\ \vdots \\ \x_c^{k+1}
    \end{bmatrix} + \begin{bmatrix}
        B_1 \\ B_2 \\ \vdots \\ 0
    \end{bmatrix} \\
     \implies &\begin{bmatrix}
        \x_1^{k+1} \\ \x_2^{k+1} \\ \vdots \\ \x_c^{k+1}
    \end{bmatrix} = \begin{bmatrix}
        1 & 0 & \ldots & 0 \\
        0 & 1 & \ldots & 0 \\
        0 & 0 & \ddots & 0 \\
        -1/n & -1/n & \ldots & 1
    \end{bmatrix}^{-1} \left(\begin{bmatrix}
        0 & 0 & \ldots & 1-\alpha A_1 \\
        0 & 0 & \ldots & 1-\alpha A_2 \\
        0 & 0 & \ddots & 1-\alpha A_i \\
        0 & 0 & \ldots & 0
    \end{bmatrix} \begin{bmatrix}
        \x_1^{k+1} \\ \x_2^{k+1} \\ \vdots \\ \x_c^{k+1}
    \end{bmatrix} + \begin{bmatrix}
        B_1 \\ B_2 \\ \vdots \\ 0
    \end{bmatrix}\right) \\
    \implies & \begin{bmatrix}
        \x_1^{k+1} \\ \x_2^{k+1} \\ \vdots \\ \x_c^{k+1}
    \end{bmatrix} = \begin{bmatrix}
        0 & 0 & \ldots & 1-\alpha A_1 \\
        0 & 0 & \ldots & 1-\alpha A_2 \\
        0 & 0 & \ddots & 1-\alpha A_i \\
        0 & 0 & \ldots & 1 -\frac{\alpha}{n}\sum_{i=1}^n A_i
    \end{bmatrix} \begin{bmatrix}
        \x_1^{k+1} \\ \x_2^{k+1} \\ \vdots \\ \x_c^{k+1}
    \end{bmatrix} + \begin{bmatrix}
        1 & 0 & \ldots & 0 \\
        0 & 1 & \ldots & 0 \\
        0 & 0 & \ddots & 0 \\
        1/n & 1/n & \ldots & 1
    \end{bmatrix} \begin{bmatrix}
        B_1 \\ B_2 \\ \vdots \\ 0
    \end{bmatrix}
\end{align}

The iterative matrix for CGC is defined as:
\begin{equation}
    G_{cgc} = \begin{bmatrix}
        0 & 0 & \ldots & 1-\alpha A_1 \\
        0 & 0 & \ldots & 1-\alpha A_2 \\
        0 & 0 & \ddots & 1-\alpha A_i \\
        0 & 0 & \ldots & 1 -\frac{\alpha}{n}\sum_{i=1}^n A_i
    \end{bmatrix}.
\end{equation}
The eigenvalues of $G_
{cgc}$ are: 
\begin{equation}
    \lambda_{cgc} = [0, 0, \ldots, 0, eig(1-\frac{\alpha}{n}\sum_{i}^{n} A_i)],
\end{equation}
and the spectral radius, $\rho(G_{cgc})$, is:
\begin{equation}
    \rho(G_{cgc}) = max(eig(1-\frac{\alpha}{n}\sum_{i=1}^{n} A_i)).
\end{equation}
Therefore, from a G-S approach \cite{white2012relaxation}, the convergence rate of CGC is bounded by:
\begin{equation}
    O(e^{\rho(G_{one})k}) = O(e^{max(eig(1-\frac{\alpha}{n}\sum_{i=1}^{n}A_i))k}).
\end{equation}

\subsection{Convergence Rate of DANE}
DANE uses a Hessian approximation to define a consensus update. With a fixed step-size of $\alpha$, the DANE update step is:
\begin{align}
    \x_i^{k+1} &= \x_c^{k} - \alpha (A_i + \mu I)^{-1}(A_i \x_c^k + B_i) \\
    & = (I - \alpha (A_i + \mu I)^{-1}A_i)\x_c^k - \alpha(A_i+\mu I)^{-1}B_i \\
    \x_c^{k+1} &= \frac{1}{n}\sum_{i=1}^{n} \x_i^{k+1}.
\end{align}

The update in a matrix form is:
\begin{align}
    \begin{bmatrix}
        1 & 0 & \ldots & 0 \\
        0 & 1 & \ldots & 0 \\
        0 & 0 & \ddots & 0 \\
        -1/n & -1/n & \ldots & 1
    \end{bmatrix} \begin{bmatrix}
        \x_1 ^{k+1} \\
        \x_2^{k+1} \\
        \vdots \\
        \x_c^{k+1}
    \end{bmatrix} = \begin{bmatrix}
        0 & 0 & \ldots & (I - \alpha (A_1 + \mu I)^{-1}A_1)\\
        0 & 0 & \ldots & (I - \alpha (A_2 + \mu I)^{-1}A_2) \\
        0 & 0 & \ddots & (I - \alpha (A_n + \mu I)^{-1}A_n) \\
        0 & 0 & \ldots &0
    \end{bmatrix} \begin{bmatrix}
        \x_1^k \\
        \x_2^k \\
        \vdots \\
        \x_c^k
    \end{bmatrix}
    + \begin{bmatrix}
        (I - \alpha (A_1 + \mu I)^{-1})B_1 \\
        (I - \alpha (A_2 + \mu I)^{-1})B_2 \\
        \vdots \\
        0
    \end{bmatrix} \\
    \implies  \begin{bmatrix}
        \x_1 ^{k+1} \\
        \x_2^{k+1} \\
        \vdots \\
        \x_c^{k+1}
    \end{bmatrix} = \begin{bmatrix}
        1 & 0 & \ldots & 0 \\
        0 & 1 & \ldots & 0 \\
        0 & 0 & \ddots & 0 \\
        -1/n & -1/n & \ldots & 1
    \end{bmatrix}^{-1} \left ( \begin{bmatrix}
        0 & 0 & \ldots & (I - \alpha (A_1 + \mu I)^{-1}A_1)\\
        0 & 0 & \ldots & (I - \alpha (A_2 + \mu I)^{-1}A_2) \\
        0 & 0 & \ddots & (I - \alpha (A_n + \mu I)^{-1}A_n) \\
        0 & 0 & \ldots &0
    \end{bmatrix} \begin{bmatrix}
        \x_1^k \\
        \x_2^k \\
        \vdots \\
        \x_c^k
    \end{bmatrix}
    + \begin{bmatrix}
        (I - \alpha (A_1 + \mu I)^{-1})B_1 \\
        (I - \alpha (A_2 + \mu I)^{-1})B_2 \\
        \vdots \\
        0
    \end{bmatrix}\right ) \\
    \implies = \begin{bmatrix}
        0 & 0 & \ldots & (I - \alpha (A_1 + \mu I)^{-1}A_1)\\
        0 & 0 & \ldots & (I - \alpha (A_2 + \mu I)^{-1}A_2) \\
        0 & 0 & \ddots & (I - \alpha (A_n + \mu I)^{-1}A_n) \\
        0 & 0 & \ldots & 1-\frac{\alpha}{n} \sum_{i=1}^{n} (A_i+\mu I)^{-1}A_i
    \end{bmatrix} \begin{bmatrix}
        \x_1^k \\
        \x_2^k \\
        \vdots \\
        \x_c^k
    \end{bmatrix} + \begin{bmatrix}
        1 & 0 & \ldots & 0 \\
        0 & 1 & \ldots & 0 \\
        0 & 0 & \ddots & 0 \\
        -1/n & -1/n & \ldots & 1
    \end{bmatrix}^{-1} \begin{bmatrix}
        (I - \alpha (A_1 + \mu I)^{-1})B_1 \\
        (I - \alpha (A_2 + \mu I)^{-1})B_2 \\
        \vdots \\
        0
    \end{bmatrix}.
\end{align}

The iterative matrix for DANE is defined as:
\begin{equation}
    G_{DANE} = \begin{bmatrix}
        0 & 0 & \ldots & (I - \alpha (A_1 + \mu I)^{-1}A_1)\\
        0 & 0 & \ldots & (I - \alpha (A_2 + \mu I)^{-1}A_2) \\
        0 & 0 & \ddots & (I - \alpha (A_n + \mu I)^{-1}A_n) \\
        0 & 0 & \ldots & 1-\frac{\alpha}{n} \sum_{i=1}^{n} (A_i+\mu I)^{-1}A_i
    \end{bmatrix}.
\end{equation}

and the spectral radius of $G_{DANE}$ is:
\begin{equation}
    \rho(G_{DANE}) = \max eig(1-\frac{\alpha}{n} \sum_{i=1}^{n} (A_i+\mu I)^{-1}A_i)
\end{equation}

The convergence rate is therefore:
\begin{equation}
        O(e^{\rho(G_{DANE})k}) = O(e^{\max eig(1-\frac{\alpha}{n} \sum_{i=1}^{n} (A_i+\mu I)^{-1}A_i) k}).
\end{equation}

\subsection{Convergence Rate of ADMM}

ADMM introduces a set of dual-variables, $\lambda_i$ which are updated at each iteration. The ADMM update step with a step-size of $\alpha$ and constant $\eta$:
\begin{align}
    \x_i^{k+1} = \x_i^{k} - \alpha (A_i \x_i^k + B_i + \eta(\x_i^k - \x_c^k + \frac{1}{\eta}\lambda_i^k) \\
    \x_c^{k+1} = \frac{1}{n} \sum_{i=1}^{n} \x_i^{k+1} + \frac{1}{\eta} \lambda_i^k \\
    \lambda_i^{k+1} = \frac{1}{n}\lambda_i^k + \x_i^{k+1} - \x_c^{k+1}.
\end{align}

In matrix form, the update is:
\begin{multline}
        \begin{bmatrix}
        1 & 0 & \ldots & 0 & 0 & 0 &\ldots\\
        0 & 1 & \ldots & 0 & 0 & 0 & \ldots \\
        0 & 0 & \ddots & 0 & 0 & 0 & \ldots \\
        -1/n& -1/n & \ldots 1 & 0 & 0 &\ldots \\
        -1 & 0 & \ldots & 1 & 1/\eta & 0 & \ldots \\
        0 & -1 & \ldots & 1 & 0 & 1/\eta & \ldots\\
        0 & 0 & \ddots & 1 & 0 & 0 & \ddots
    \end{bmatrix} \begin{bmatrix}
        \x_1^{k+1}\\
        \x_2^{k+1}\\
        \vdots \\
        \x_c^{k+1} \\
        \lambda_1^{k+1} \\
        \lambda_2^{k+1} \\
        \vdots
    \end{bmatrix} = \\ \begin{bmatrix}
        1-\alpha (A_1+\eta) & 0 & \ldots & -\alpha \eta & -\alpha & 0 & \ldots \\
        0 & 1-\alpha(A_2+\eta) & \ldots & -\alpha\eta & 0 &-\alpha & \ldots \\
        0 & 0 & \ddots & -\alpha\eta & 0 & 0& \ddots \\
        0 & 0 & \ldots &0 & 1/n\eta & 1/n\eta & \ldots \\
        0 & 0 & \ldots & 0 & 1/\eta & 0 & \ldots \\
        0 & 0 & \ldots & 0 & 0 &1\eta &\ldots \\
        0 & 0 & \ldots & 0 & 0 & 0 &\ddots
    \end{bmatrix}\begin{bmatrix}
        \x_1^k\\
        \x_2^k\\
        \vdots \\
        \x_c^k \\
        \lambda_1^k\\
        \lambda_2^k \\
        \vdots
    \end{bmatrix} + \begin{bmatrix}
        -\alpha B_1 \\
        -\alpha B_2 \\
        \vdots \\
        0 \\
        0 \\
        0 \\
        \vdots
    \end{bmatrix} \label{eq:admm_update_matrix}
\end{multline}

We define matrices $G_1$ and $G_2$ as the matrices on the left and right hand side of \eqref{eq:admm_update_matrix}:
\begin{equation}
    G_1 \equiv \begin{bmatrix}
        1 & 0 & \ldots & 0 & 0 & 0 &\ldots\\
        0 & 1 & \ldots & 0 & 0 & 0 & \ldots \\
        0 & 0 & \ddots & 0 & 0 & 0 & \ldots \\
        -1/n& -1/n & \ldots 1 & 0 & 0 &\ldots \\
        -1 & 0 & \ldots & 1 & 1/\eta & 0 & \ldots \\
        0 & -1 & \ldots & 1 & 0 & 1/\eta & \ldots\\
        0 & 0 & \ddots & 1 & 0 & 0 & \ddots
    \end{bmatrix} \;\;, G_2 \equiv \begin{bmatrix}
        1-\alpha (A_1+\eta) & 0 & \ldots & -\alpha \eta & -\alpha & 0 & \ldots \\
        0 & 1-\alpha(A_2+\eta) & \ldots & -\alpha\eta & 0 &-\alpha & \ldots \\
        0 & 0 & \ddots & -\alpha\eta & 0 & 0& \ddots \\
        0 & 0 & \ldots &0 & 1/n\eta & 1/n\eta & \ldots \\
        0 & 0 & \ldots & 0 & 1/\eta & 0 & \ldots \\
        0 & 0 & \ldots & 0 & 0 &1\eta &\ldots \\
        0 & 0 & \ldots & 0 & 0 & 0 &\ddots
    \end{bmatrix}
\end{equation}

where $G_1$ is a block-lower-diagonal-triangular matrix with eigenvalues:
\begin{equation}
    \lambda_{G_1} = [1,1/\eta],
\end{equation}
and $G_2$ is a block-upper-diagonal-triangular matrix with eigenvalues:
\begin{equation}
    \lambda_{G_2} = [eig(1-\alpha(A_1+\eta)), eig(1-\alpha(A_2+\eta)), \ldots, 1/\eta].
\end{equation}

For values of $\eta>1$, the spectral radius of $G_1$ is:
\begin{equation}
    \rho(G_1) = 1,
\end{equation}
and the spectral radius of $G_2$ is:
\begin{equation}
    \rho(G_2) = \max eig(1-\alpha(A_i+\eta I))
\end{equation}

The iterative matrix for ADMM, $G_{ADMM}$ is defined as:
\begin{equation}
    G_{ADMM} = G_1^{-1}G_2,
\end{equation}
which has spectral radius bounded by:
\begin{align}
    \rho(G_{ADMM}) &\leq \rho(G_1^{-1})\rho(G_2) \\
    & \leq \max eig(1-\alpha(A_i+\eta I)).
\end{align}

Therefore the convergence rate of ADMM is upper bounded by:
\begin{equation}
    O(e^{\rho(G_{ADMM})k}) \leq O(e^{\max eig(1-\alpha(A_i+\eta I))k}).
\end{equation}

\printbibliography

@article{alberto_numerical_methods,
author = {De Micheli, Giovanni and Sangiovanni-Vincentelli, Alberto},
title = {Characterization of integration algorithms for the timing analysis of mos vlsi circuits},
journal = {International Journal of Circuit Theory and Applications},
volume = {10},
number = {4},
pages = {299-309},
doi = {https://doi.org/10.1002/cta.4490100402},
url = {https://onlinelibrary.wiley.com/doi/abs/10.1002/cta.4490100402},
eprint = {https://onlinelibrary.wiley.com/doi/pdf/10.1002/cta.4490100402},
year = {1982}
}

@article{yuejie_li2020communication,
  title={Communication-efficient distributed optimization in networks with gradient tracking and variance reduction},
  author={Li, Boyue and Cen, Shicong and Chen, Yuxin and Chi, Yuejie},
  journal={The Journal of Machine Learning Research},
  volume={21},
  number={1},
  pages={7331--7381},
  year={2020},
  publisher={JMLRORG}
}

@article{xiao2004fast,
  title={Fast linear iterations for distributed averaging},
  author={Xiao, Lin and Boyd, Stephen},
  journal={Systems \& Control Letters},
  volume={53},
  number={1},
  pages={65--78},
  year={2004},
  publisher={Elsevier}
}

@article{admm_boyd2011distributed,
  title={Distributed optimization and statistical learning via the alternating direction method of multipliers},
  author={Boyd, Stephen and Parikh, Neal and Chu, Eric and Peleato, Borja and Eckstein, Jonathan and others},
  journal={Foundations and Trends{\textregistered} in Machine learning},
  volume={3},
  number={1},
  pages={1--122},
  year={2011},
  publisher={Now Publishers, Inc.}
}

@inproceedings{dane_shamir2014communication,
  title={Communication-efficient distributed optimization using an approximate newton-type method},
  author={Shamir, Ohad and Srebro, Nati and Zhang, Tong},
  booktitle={International conference on machine learning},
  pages={1000--1008},
  year={2014},
  organization={PMLR}
}

@ARTICLE{9253684,
  author={Zhang, Mengyao and Liu, Xinzhi and Liu, Jun},
  journal={IEEE Control Systems Letters}, 
  title={Convergence Analysis of a Continuous-Time Distributed Gradient Descent Algorithm}, 
  year={2021},
  volume={5},
  number={4},
  pages={1339-1344},
  doi={10.1109/LCSYS.2020.3037038}}

@ARTICLE{6578120,
  author={Gharesifard, Bahman and Cortés, Jorge},
  journal={IEEE Transactions on Automatic Control}, 
  title={Distributed Continuous-Time Convex Optimization on Weight-Balanced Digraphs}, 
  year={2014},
  volume={59},
  number={3},
  pages={781-786},
  doi={10.1109/TAC.2013.2278132}}

@article{kia2015distributed,
  title={Distributed convex optimization via continuous-time coordination algorithms with discrete-time communication},
  author={Kia, Solmaz S and Cort{\'e}s, Jorge and Mart{\'\i}nez, Sonia},
  journal={Automatica},
  volume={55},
  pages={254--264},
  year={2015},
  publisher={Elsevier}
}

@INPROCEEDINGS{5706956,
  author={Wang, Jing and Elia, Nicola},
  booktitle={2010 48th Annual Allerton Conference on Communication, Control, and Computing (Allerton)}, 
  title={Control approach to distributed optimization}, 
  year={2010},
  volume={},
  number={},
  pages={557-561},
  doi={10.1109/ALLERTON.2010.5706956}}

@INPROCEEDINGS{7554656,
  author={Xie, Yijing and Lin, Zongli},
  booktitle={2016 35th Chinese Control Conference (CCC)}, 
  title={Global optimal consensus of multi-agent systems with bounded controls}, 
  year={2016},
  volume={},
  number={},
  pages={8166-8171},
  doi={10.1109/ChiCC.2016.7554656}}

@book{lawrence1995electronic,
  title={Electronic Circuit and System Simulation Methods},
  author={Lawrence T.. Pillage and Rohrer, Ronald A and Visweswariah, Chandramouli},
  year={1995},
  publisher={McGraw-Hill}
}

@article{di2016next,
  title={Next: In-network nonconvex optimization},
  author={Di Lorenzo, Paolo and Scutari, Gesualdo},
  journal={IEEE Transactions on Signal and Information Processing over Networks},
  volume={2},
  number={2},
  pages={120--136},
  year={2016},
  publisher={IEEE}
}

@article{shi2015extra,
  title={Extra: An exact first-order algorithm for decentralized consensus optimization},
  author={Shi, Wei and Ling, Qing and Wu, Gang and Yin, Wotao},
  journal={SIAM Journal on Optimization},
  volume={25},
  number={2},
  pages={944--966},
  year={2015},
  publisher={SIAM}
}

@book{centralized_gradient_descent,
  title={Parallel and distributed computation: numerical methods},
  author={Bertsekas, Dimitri and Tsitsiklis, John},
  year={2015},
  publisher={Athena Scientific}
}

@book{behrman1998efficient,
  title={An efficient gradient flow method for unconstrained optimization},
  author={Behrman, William},
  year={1998},
  publisher={stanford university}
}

@article{attouch1996dynamical,
  title={A dynamical approach to convex minimization coupling approximation with the steepest descent method},
  author={Attouch, Hedy and Cominetti, Roberto},
  journal={Journal of Differential Equations},
  volume={128},
  number={2},
  pages={519--540},
  year={1996},
  publisher={Elsevier}
}

@article{wilson2021lyapunov,
  title={A lyapunov analysis of accelerated methods in optimization},
  author={Wilson, Ashia C and Recht, Ben and Jordan, Michael I},
  journal={Journal of Machine Learning Research},
  volume={22},
  number={113},
  pages={1--34},
  year={2021}
}

@article{polyak2017lyapunov,
  title={Lyapunov functions: An optimization theory perspective},
  author={Polyak, Boris and Shcherbakov, Pavel},
  journal={IFAC-PapersOnLine},
  volume={50},
  number={1},
  pages={7456--7461},
  year={2017},
  publisher={Elsevier}
}

@InProceedings{muehlebach19a,
  title = 	 {A Dynamical Systems Perspective on {N}esterov Acceleration},
  author =       {Muehlebach, Michael and Jordan, Michael},
  booktitle = 	 {Proceedings of the 36th International Conference on Machine Learning},
  pages = 	 {4656--4662},
  year = 	 {2019},
  editor = 	 {Chaudhuri, Kamalika and Salakhutdinov, Ruslan},
  volume = 	 {97},
  series = 	 {Proceedings of Machine Learning Research},
  month = 	 {09--15 Jun},
  publisher =    {PMLR},
  url = 	 {https://proceedings.mlr.press/v97/muehlebach19a.html}
}

@article{yang2019survey,
  title={A survey of distributed optimization},
  author={Yang, Tao and Yi, Xinlei and Wu, Junfeng and Yuan, Ye and Wu, Di and Meng, Ziyang and Hong, Yiguang and Wang, Hong and Lin, Zongli and Johansson, Karl H},
  journal={Annual Reviews in Control},
  volume={47},
  pages={278--305},
  year={2019},
  publisher={Elsevier}
}

@article{swenson2021distributed,
  title={Distributed gradient flow: Nonsmoothness, nonconvexity, and saddle point evasion},
  author={Swenson, Brian and Murray, Ryan and Poor, H Vincent and Kar, Soummya},
  journal={IEEE Transactions on Automatic Control},
  year={2021},
  publisher={IEEE}
}

@article{gharesifard2013distributed,
  title={Distributed continuous-time convex optimization on weight-balanced digraphs},
  author={Gharesifard, Bahman and Cort{\'e}s, Jorge},
  journal={IEEE Transactions on Automatic Control},
  volume={59},
  number={3},
  pages={781--786},
  year={2013},
  publisher={IEEE}
}

@inproceedings{xin2019distributed,
  title={Distributed stochastic optimization with gradient tracking over strongly-connected networks},
  author={Xin, Ran and Sahu, Anit Kumar and Khan, Usman A and Kar, Soummya},
  booktitle={2019 IEEE 58th Conference on Decision and Control (CDC)},
  pages={8353--8358},
  year={2019},
  organization={IEEE}
}

@article{scutari2019distributed,
  title={Distributed nonconvex constrained optimization over time-varying digraphs},
  author={Scutari, Gesualdo and Sun, Ying},
  journal={Mathematical Programming},
  volume={176},
  pages={497--544},
  year={2019},
  publisher={Springer}
}

@article{qu2017harnessing,
  title={Harnessing smoothness to accelerate distributed optimization},
  author={Qu, Guannan and Li, Na},
  journal={IEEE Transactions on Control of Network Systems},
  volume={5},
  number={3},
  pages={1245--1260},
  year={2017},
  publisher={IEEE}
}

@article{nedic2017achieving,
  title={Achieving geometric convergence for distributed optimization over time-varying graphs},
  author={Nedic, Angelia and Olshevsky, Alex and Shi, Wei},
  journal={SIAM Journal on Optimization},
  volume={27},
  number={4},
  pages={2597--2633},
  year={2017},
  publisher={SIAM}
}

@ARTICLE{liu_small_gain,
  author={Liu, Tengfei and Qin, Zhengyan and Hong, Yiguang and Jiang, Zhong-Ping},
  journal={IEEE Transactions on Automatic Control}, 
  title={Distributed Optimization of Nonlinear Multiagent Systems: A Small-Gain Approach}, 
  year={2022},
  volume={67},
  number={2},
  pages={676-691},
  doi={10.1109/TAC.2021.3053549}}

@INPROCEEDINGS{kvaternik,
  author={Kvaternik, Karla and Pavel, Lacra},
  booktitle={International Conference on NETwork Games, Control and Optimization (NetGCooP 2011)}, 
  title={Lyapunov analysis of a distributed optimization scheme}, 
  year={2011},
  volume={},
  number={},
  pages={1-5},
  doi={}}

@ARTICLE{Liu_second_order,
  author={Liu, Qingshan and Wang, Jun},
  journal={IEEE Transactions on Automatic Control}, 
  title={A Second-Order Multi-Agent Network for Bound-Constrained Distributed Optimization}, 
  year={2015},
  volume={60},
  number={12},
  pages={3310-3315},
  doi={10.1109/TAC.2015.2416927}}

@ARTICLE{sakurama_distributed,
  author={Sakurama, Kazunori and Azuma, Shun-ichi and Sugie, Toshiharu},
  journal={IEEE Transactions on Automatic Control}, 
  title={Distributed Controllers for Multi-Agent Coordination Via Gradient-Flow Approach}, 
  year={2015},
  volume={60},
  number={6},
  pages={1471-1485},
  doi={10.1109/TAC.2014.2374951}}

@article{rahili2016distributed,
  title={Distributed continuous-time convex optimization with time-varying cost functions},
  author={Rahili, Salar and Ren, Wei},
  journal={IEEE Transactions on Automatic Control},
  volume={62},
  number={4},
  pages={1590--1605},
  year={2016},
  publisher={IEEE}
}

@inproceedings{pilloni2016discontinuous,
  title={A discontinuous algorithm for distributed convex optimization},
  author={Pilloni, Alessandro and Pisano, Alessandro and Franceschelli, Mauro and Usai, Elio},
  booktitle={2016 14th International Workshop on Variable Structure Systems (VSS)},
  pages={22--27},
  year={2016},
  organization={IEEE}
}

@article{white2012relaxation,
  title={Relaxation techniques for the simulation of VLSI circuits},
  author={White, Jacob K and Sangiovanni-Vincentelli, Alberto L},
  year={2012},
  publisher={Springer Science \& Business Media}
}

@article{erseghe2014distributed,
  title={Distributed optimal power flow using ADMM},
  author={Erseghe, Tomaso},
  journal={IEEE transactions on power systems},
  volume={29},
  number={5},
  pages={2370--2380},
  year={2014},
  publisher={IEEE}
}

@article{mhanna2018adaptive,
  title={Adaptive ADMM for distributed AC optimal power flow},
  author={Mhanna, Sleiman and Verbi{\v{c}}, Gregor and Chapman, Archie C},
  journal={IEEE Transactions on Power Systems},
  volume={34},
  number={3},
  pages={2025--2035},
  year={2018},
  publisher={IEEE}
}

@article{wang2016fully,
  title={A fully-decentralized consensus-based ADMM approach for DC-OPF with demand response},
  author={Wang, Yamin and Wu, Lei and Wang, Shouxiang},
  journal={IEEE Transactions on Smart Grid},
  volume={8},
  number={6},
  pages={2637--2647},
  year={2016},
  publisher={IEEE}
}

@inproceedings{nishihara2015general,
  title={A general analysis of the convergence of ADMM},
  author={Nishihara, Robert and Lessard, Laurent and Recht, Ben and Packard, Andrew and Jordan, Michael},
  booktitle={International conference on machine learning},
  pages={343--352},
  year={2015},
  organization={PMLR}
}

@article{ghadimi2014optimal,
  title={Optimal parameter selection for the alternating direction method of multipliers (ADMM): quadratic problems},
  author={Ghadimi, Euhanna and Teixeira, Andr{\'e} and Shames, Iman and Johansson, Mikael},
  journal={IEEE Transactions on Automatic Control},
  volume={60},
  number={3},
  pages={644--658},
  year={2014},
  publisher={IEEE}
}

@inproceedings{xu2017adaptive,
  title={Adaptive consensus ADMM for distributed optimization},
  author={Xu, Zheng and Taylor, Gavin and Li, Hao and Figueiredo, M{\'a}rio AT and Yuan, Xiaoming and Goldstein, Tom},
  booktitle={International Conference on Machine Learning},
  pages={3841--3850},
  year={2017},
  organization={PMLR}
}

@article{wang2019global,
  title={Global convergence of ADMM in nonconvex nonsmooth optimization},
  author={Wang, Yu and Yin, Wotao and Zeng, Jinshan},
  journal={Journal of Scientific Computing},
  volume={78},
  pages={29--63},
  year={2019},
  publisher={Springer}
}

@article{he2016convergence,
  title={Convergence study on the symmetric version of ADMM with larger step sizes},
  author={He, Bingsheng and Ma, Feng and Yuan, Xiaoming},
  journal={SIAM journal on imaging sciences},
  volume={9},
  number={3},
  pages={1467--1501},
  year={2016},
  publisher={SIAM}
}

@inproceedings{franca2018admm,
  title={ADMM and accelerated ADMM as continuous dynamical systems},
  author={Franca, Guilherme and Robinson, Daniel and Vidal, Rene},
  booktitle={International Conference on Machine Learning},
  pages={1559--1567},
  year={2018},
  organization={PMLR}
}

@inproceedings{mokhtari2015approximate,
  title={An approximate Newton method for distributed optimization},
  author={Mokhtari, Aryan and Ling, Qing and Ribeiro, Alejandro},
  booktitle={2015 IEEE International Conference on Acoustics, Speech and Signal Processing (ICASSP)},
  pages={2959--2963},
  year={2015},
  organization={IEEE}
}

@article{mokhtari2016network,
  title={Network Newton distributed optimization methods},
  author={Mokhtari, Aryan and Ling, Qing and Ribeiro, Alejandro},
  journal={IEEE Transactions on Signal Processing},
  volume={65},
  number={1},
  pages={146--161},
  year={2016},
  publisher={IEEE}
}

@article{tutunov2019distributed,
  title={Distributed newton method for large-scale consensus optimization},
  author={Tutunov, Rasul and Bou-Ammar, Haitham and Jadbabaie, Ali},
  journal={IEEE Transactions on Automatic Control},
  volume={64},
  number={10},
  pages={3983--3994},
  year={2019},
  publisher={IEEE}
}

@article{bajovic2017newton,
  title={Newton-like method with diagonal correction for distributed optimization},
  author={Bajovic, Dragana and Jakovetic, Dusan and Krejic, Natasa and Jerinkic, Natasa Krklec},
  journal={SIAM Journal on Optimization},
  volume={27},
  number={2},
  pages={1171--1203},
  year={2017},
  publisher={SIAM}
}

@inproceedings{scaman2017optimal,
  title={Optimal algorithms for smooth and strongly convex distributed optimization in networks},
  author={Scaman, Kevin and Bach, Francis and Bubeck, S{\'e}bastien and Lee, Yin Tat and Massouli{\'e}, Laurent},
  booktitle={international conference on machine learning},
  pages={3027--3036},
  year={2017},
  organization={PMLR}
}

@article{uribe2017optimal,
  title={Optimal algorithms for distributed optimization},
  author={Uribe, C{\'e}sar A and Lee, Soomin and Gasnikov, Alexander and Nedi{\'c}, Angelia},
  journal={arXiv preprint arXiv:1712.00232},
  year={2017}
}

@article{wai2018multi,
  title={Multi-agent reinforcement learning via double averaging primal-dual optimization},
  author={Wai, Hoi-To and Yang, Zhuoran and Wang, Zhaoran and Hong, Mingyi},
  journal={Advances in Neural Information Processing Systems},
  volume={31},
  year={2018}
}

@article{gisetteData,
  title={Result analysis of the NIPS 2003 feature selection challenge},
  author={Guyon, Isabelle and Gunn, Steve and Ben-Hur, Asa and Dror, Gideon},
  journal={Advances in neural information processing systems},
  volume={17},
  year={2004}
}

@misc{go2, note = {``Challenge 2'', \url{https://gocompetition.energy.gov/challenges/challenge-2}}}

@article{imb,
  title={Incremental model building homotopy approach for solving exact ac-constrained optimal power flow},
  author={Pandey, Amritanshu and Agarwal, Aayushya and Pileggi, Larry},
  journal={arXiv preprint arXiv:2011.00587},
  year={2020}
}

@misc{boyd2021distributed,
  title={Distributed optimization: Analysis and synthesis via circuits},
  author={Boyd, Stephen},
  year={2021},
  publisher={May}
}

@misc{agarwal2023equivalent,
      title={An Equivalent Circuit Workflow for Unconstrained Optimization}, 
      author={Aayushya Agarwal and Carmel Fiscko and Soummya Kar and Larry Pileggi and Bruno Sinopoli},
      year={2023},
      eprint={2305.14061},
      archivePrefix={arXiv},
      primaryClass={math.OC}
}

\end{document}